\documentclass[journal,onecolumn]{IEEEtran}
\usepackage{cite}
\usepackage{amsfonts}
\usepackage{mathrsfs}
\usepackage{graphicx}
\usepackage{amssymb}
\usepackage{latexsym}
\usepackage{amsmath}
\usepackage{stfloats}
\usepackage{cases}
\usepackage{setspace}
\usepackage{bigstrut}
\usepackage{bm}
\usepackage{url}
\usepackage{color}
\usepackage{algorithm}
\usepackage{algorithmic}

\renewcommand{\baselinestretch}{1.2}

\newtheorem{theorem}{Theorem}
\newtheorem{assumption}{Assumption}

\begin{document}

\title{FDD Massive MIMO Based on Efficient Downlink Channel Reconstruction}

\author{Yu~Han$^\ast$, Qi Liu$^\ast$, Chao-Kai~Wen$^\dagger$, Shi~Jin$^\ast$, and~Kai-Kit~Wong$^\S$\\
$^\ast$National Mobile Communications Research Laboratory, Southeast University, P. R. China\\
$^\dagger$Institute of Communications Engineering, National Sun Yat-sen University, Taiwan\\
$^\S$University College London, London, United Kingdom\\
Email: \{hanyu,qiliu\}@seu.edu.cn, chaokai.wen@mail.nsysu.edu.tw, \\jinshi@seu.edu.cn, kai-kit.wong@ucl.ac.uk}

\maketitle

\vspace{-5mm}

\begin{abstract}
Massive multiple-input multiple-output (MIMO) systems deploying a large number of antennas at the base station considerably increase the spectrum efficiency by serving multiple users simultaneously without causing severe interference. However, the advantage relies on the availability of the downlink channel state information (CSI) of multiple users, which is still a challenge in frequency-division-duplex transmission systems. This paper aims to solve this problem by developing a full transceiver framework that includes downlink channel training (or estimation), CSI feedback, and channel reconstruction schemes. Our framework provides accurate reconstruction results for multiple users with small amounts of training and feedback overhead. Specifically, we first develop an enhanced Newtonized orthogonal matching pursuit (eNOMP) algorithm to extract the frequency-independent parameters (i.e., downtilts, azimuths, and delays) from the uplink. Then, by leveraging the information from these frequency-independent parameters, we develop an efficient downlink training scheme to estimate the downlink channel gains for multiple users. This training scheme offers an acceptable estimation error rate of the gains with a limited pilot amount. Numerical results verify the precision of the eNOMP algorithm and demonstrate that the sum-rate performance of the system using the reconstructed downlink channel can approach that of the system using perfect CSI.
\end{abstract}

\begin{keywords}
Downlink channel reconstruction, FDD massive MIMO, multiuser transmission.
\end{keywords}

\section{Introduction}\label{Sec:Introduction}

Massive multiple-input multiple-output (MIMO) is a key enabler of the fifth-generation and future mobile communication networks \cite{Andrews2014,Larsson2014,Agiwal2016}. Large-scale antenna arrays are equipped at the base stations (BSs) to fully exploit the spatial degrees of freedom, for providing huge room for spatial division multiplexing \cite{Adhikary2013,Mohammed2013,Yu2018}. Multiple users can be served by the BS on the same time-frequency resource block, and the spatial multiplexing dimension can be further expanded by scaling up the array at the BS. The large array is usually structured in a planar or circular topology to exploit horizontal and vertical spaces and realize three dimensional (3D) MIMO techniques. A beam formed by the array can flexibly target at any direction in the 3D space depending on practical requirements \cite{Nam2013,Halbauer2013,Li2016}. We can also design beam weights to produce a set of spatially orthogonal beams using the large array and transmit multiple data streams on these beams without causing interference \cite{Sun2015}. These advantages result in the high sum-rate performance of multiuser massive MIMO systems.

A prerequisite to gain these advantages is the acquisition of the channel state information (CSI). Due to the lack of uplink-downlink reciprocity in frequency-division-duplex (FDD) systems, downlink training-then-feedback is a typical solution for downlink channel estimation. In the fourth-generation era and before, the number of BS antennas is relatively small. Thus, downlink CSI can be easily acquired by sending orthogonal downlink pilots, applying channel estimation at the user side, and finally feeding the estimates back to the BS. However, in massive MIMO systems, using completely orthogonal downlink pilots and sending back high-dimensional complex channel matrices are impractical. Obtaining downlink CSI at the BS side becomes a bottleneck in FDD massive MIMO systems. Therefore, researchers have been searching for new solutions to obtain downlink CSI and design corresponding transmission schemes.

{\bf Related Work:}
In this area, some methods followed the traditional approach by transmitting downlink pilots and sending back the estimates to the BS. For example, compressed sensing was introduced to estimate the sparse channel through a small amount of downlink measurements \cite{Gao2015,Rao2014}. However, comprehensive signal processing was conducted at the user side, which raised an exorbitant requirement on the capability of the user equipment. Moreover, time-correlation of the wireless channel can be utilized in downlink training and feedback phases. In \cite{Choi2014}, the user estimated downlink CSI based on the currently received downlink pilots and the estimated CSI obtained at the previous moment. Before feeding back the estimates, \cite{Choi2015} and \cite{Shen2017} proposed to quantize the channel based on previous results within the coherence time by using a trellis-extended codebook and an angle of departure-adaptive subspace codebook. However, these methods \cite{Choi2014,Choi2015,Shen2017} relied on the accuracy of the initial estimates.

In recent years, the idea of using the spatial reciprocity between uplink and downlink has attracted increasing attention. Many efforts have suggested to acquire downlink CSI by using the information obtained from the uplink. Existing work generally aims to obtain two types of downlink CSI. The first type is partial CSI, such as the spatial information or the reduced dimensional channel. For example, only the angles of propagation paths are estimated during the training phase. Alternatively, channel sparsity in beamspace is utilized and spatial angle estimation is translated to the search for the non-zero elements in the beamspace channel. In this context, only spatial directions or beam indices are known at the BS, and user scheduling is required to avoid spatial overlapping among different users \cite{Sun2015,Han2017}. The second type is full CSI, which describes the full-dimensional channel and contains the complete information in the propagation environment. Full CSI is usually obtained by channel estimation or reconstruction scheme. With full CSI, the BS can serve a large number of users simultaneously by using precoding to eliminate the interference. The BS also can conduct a comprehensive user scheduling scheme to maximize the sum-rate performance and fully exploit the spatial multiplexing gain.

Most recent works aim to acquire the full CSI \cite{Ugurlu2016,Xie2018,Haghighatshoar2018,Khalilsarai2018,Xie2017}. Focusing on the clustering channel that covers a continuous angular region, \cite{Xie2018,Haghighatshoar2018,Khalilsarai2018} suggested to estimate the channel based on the downlink channel covariance matrix that describes the angular domain energy distribution and can be derived from its uplink version. For the limited scattering channel where multiple distinct paths exist, the authors in \cite{Xie2017} proposed a unified transmission strategy with the aid of the spatial basis expansion model. Using channel sparsity, the reduced dimensional angular-domain downlink channel from each user was first estimated and then transformed back to the full-dimensional antenna domain. These methods are based on the angular domain that uniformly over-samples the space and estimate the power distribution on these sample points. If the paths are distinguishable, then the number of sample points with distinct projected power is usually larger than that of the real paths. As a result, the training overhead used to obtain the power distribution on these sample points increases.

To save the training resources, \cite{Zhang2018,Han2018wocc,Han2018arxiv,Liu2018} considered the real paths and reconstructed the downlink channel by extracting each component path. They proposed to estimate frequency-independent parameters in the uplink. Then, only the downlink gains should be estimated through downlink training and feedback, and the amount of pilots and feedback overhead was small. The effectiveness of the proposed schemes were also demonstrated in over-the-air tests. However, \cite{Zhang2018} did not consider delays, set restriction on the estimated number of paths, and only retained four dominant paths. \cite{Han2018wocc,Han2018arxiv,Liu2018} merely provided a channel reconstruction scheme to obtain the downlink channel of a single user.
Given that \cite{Han2018wocc,Han2018arxiv,Liu2018} used \emph{user-dedicated} pilots for downlink training, if we simply extend the method to fit multiuser scenarios, then a large amount of training resource is needed to estimate downlink gains of multiuser channels, as designed in \cite{Zhang2018}. Moreover, the schemes proposed in \cite{Han2018wocc} and \cite{Han2018arxiv} are inapplicable in massive 3D-MIMO systems because they can not extract the full-dimensional spatial parameters.

{\bf Contributions:}
This paper considers the multiuser massive MIMO system using orthogonal frequency-division-multiplexing (OFDM) technique under FDD mode.
By resolving the drawbacks in \cite{Han2018wocc,Han2018arxiv,Liu2018} for multiuser scenarios, we present a full transceiver framework for downlink channel reconstruction with high reconstruction precision and low training and feedback overhead. Following the mechanism in \cite{Han2018wocc,Han2018arxiv,Liu2018}, we also reconstruct the downlink channel by extracting each component path.
Specifically, our transceiver framework combines three components as described below.
\begin{enumerate}
  \item {\bf Frequency-independent parameter extraction:} We develop an enhanced Newtonized orthogonal matching pursuit (eNOMP) algorithm that can detect the component paths from their noisy mixture and extract the gain, downtilt, azimuth, and delay of each path for massive MIMO OFDM systems. Part of this work will be shown in the conference version of this paper \cite{Liu2018}. Numerical results show that eNOMP can provide precise extraction results. eNOMP-based uplink channel reconstruction also achieves smaller mean square error (MSE) than the linear minimum MSE (LMMSE) estimation.

  \item {\bf Pilot scheduling:} After obtaining the frequency-independent parameters (i.e., downtilt, azimuth, and delay) of each path from the uplink channels, the subsequent task is to estimate the downlink gains through the downlink pilots. The best way to estimate each downlink gain is to apply beamforming in downlink training process and send the pilot aligning with each downtilt and azimuth direction (i.e., dedicated pilot). However, this approach using the dedicated pilots will result in large training overhead in the multiuser scenario. To solve this problem, we propose to co-use and remove some pilots (or beams) while slightly compromising the estimation performance.
      In particular, we introduce an approximation of the MSE for the channel gains to predict the corresponding performance degradation for determining which pilots (or beams) can be removed. We use the greedy method to exclude the unnecessary pilots.
      The proposed pilot scheduling scheme can minimize the downlink training overhead and offer an acceptable estimation error rate.

  \item {\bf Multiuser downlink gain estimation and CSI feedback:} After determining the pilots, we use them in a broadcast manner instead of the dedicated manner. That is, all the users can use the broadcast pilots. However, if each user directly sends the channel response on each corresponding pilot back to the BS, then the amount of feedback remains large. To reduce the feedback overhead, the BS sends the corresponding frequency-independent parameters to each user. Thereafter, the user can estimate the downlink gain of each path on the basis of the downtilt, azimuth, and delay, and send the estimates back to the BS.
      Given that the number of paths in each user channel is much smaller than the pilot numbers, the feedback overhead of the proposed scheme is also very small. We evaluate the multiuser sum-rate performance through theoretical analysis and simulations and observe that the sum-rate offered by the reconstruction channel remains large. This finding validates the effectiveness and efficiency of the proposed downlink channel reconstruction scheme.
\end{enumerate}

The rest of the paper is organized as follows. Section \ref{Sec:SystemModel} describes the massive MIMO-OFDM system operating in FDD mode, introduces the uplink and downlink channel models in terms of the frequency-independent parameters, and briefly outlines the proposed downlink channel reconstruction and multiuser transmission scheme. Section \ref{Sec:NOMP} presents the eNOMP algorithm, and highlights the new codebook and the updated Newton step designed for the 3D massive MIMO-OFDM scenario. Section \ref{Sec:DLtraining} presents the working principle of the proposed low-cost downlink-training strategy for the multiuser system. Section \ref{Sec:MUscheme} applies the reconstruction results to the downlink multiuser transmission scheme and analyzes the sum-rate performance. Section \ref{Sec:Results} provides the numerical results, and Section \ref{Sec:Conclusion} elaborates the conclusions.

\emph{Notations}---We denote matrices and vectors by uppercase and lowercase boldface letters, respectively. By contrast, the superscripts $(\cdot)^\dag$, $(\cdot)^{H}$, and $(\cdot)^{T}$ denote pseudo-inverse, conjugate-transpose, and transpose, respectively. We also denote $[{\bf A}]_{i,:}$ and $[{\bf A}]_{:,j}$ as the $i$th row and the $j$th column of matrix $\bf A$, and $[{\bf A}]_{i,j}$ as the $(i,j)$th entry of $\bf A$. $\otimes$ denotes taking Kronecker product. $\mathcal{R}\{\cdot\}$ represents taking the real part of a complex number, whereas $\mathbb{E}\{\cdot\}$ takes the expectation with respect to the random variables inside the brackets. $\left| \cdot \right|$ and $\left\| \cdot \right\|$ indicate taking the absolute value and modulus operations, and $\left\lfloor \cdot \right\rfloor$ and $\left\lceil \cdot \right\rceil$ imply rounding a decimal number to its nearest lower and higher integers, respectively.

\section{System Model}\label{Sec:SystemModel}

\subsection{Channel Model}

\begin{figure}
  \centering
  \includegraphics[scale=0.54]{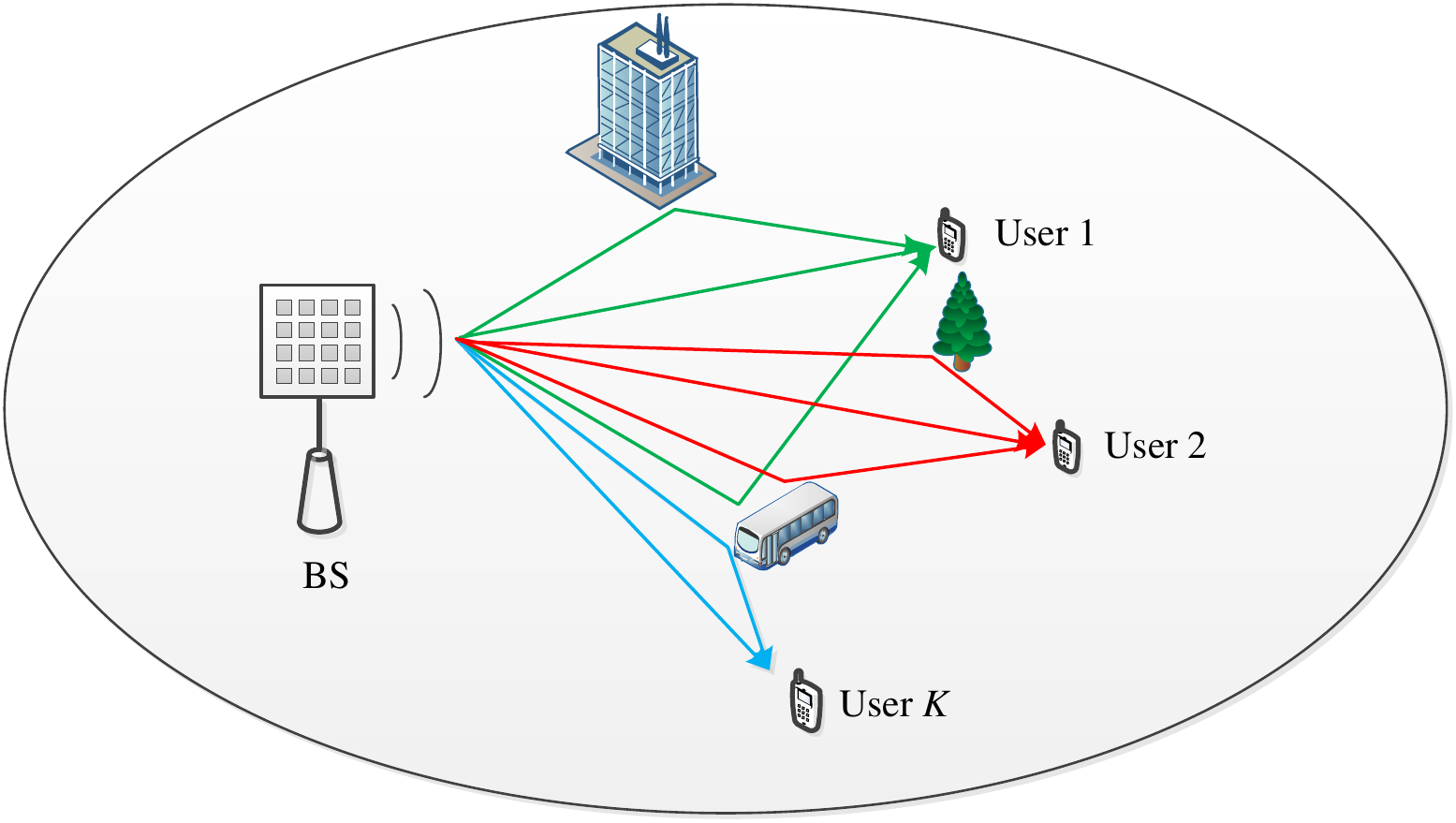}\\
  \caption{FDD massive MIMO system. BS is equipped with UPA and serves $K$ users simultaneously. Buildings, trees, and cars are all scatterers in the wireless channel.}\label{Fig:MUscenario}
\end{figure}

A single-cell massive MIMO system using FDD transmission mode and OFDM modulation is considered. We denote the uplink and downlink carrier frequencies as $f^{\rm ul}$ and $f^{\rm dl}$, respectively. We also assume that each of the uplink and downlink frequency bands has $N$ sub-carriers with spacing $\triangle f$. The BS serves $K$ users who are randomly distributed in the cell. The BS is equipped with a UPA, and each user has a single antenna. The UPA contains $M=M_hM_v$ antenna elements, including $M_h$ elements in each row and $M_v$ elements in each column, and $M\gg K$. The distance between two horizontally or vertically adjacent elements is $d=\lambda/2$, where $\lambda$ is the carrier wavelength. As shown in Fig.~\ref{Fig:MUscenario}, scatterers exist in the space, and the user channel is composed of multiple propagation paths. The wireless signal can arrive at the user side along with the line-of-sight path or be reflected by several scatterers. Different user channels may share a common scatterer and are spatially overlapped with each other. The wireless channel of each user remains constant over $T_c$ OFDM symbols.

For user $k$, when down-converted to the baseband, its uplink multipath channel is expressed as
\begin{equation}\label{Eq:ULChannel}
{\bf h}^{\rm ul}_k=\sum\limits_{l = 1}^{{L_k}} {g^{\rm ul}_{k,l} {\bf a}(\theta_{k,l},\phi_{k,l}) \otimes {\bf p} (\tau_{k,l})},
\end{equation}
where $L_k$ is the number of propagation paths of user $k$; ${\bf a} (\theta_{k,l},\phi_{k,l})$ is the steering vector of UPA with
$\theta_{k,l} \in [-\pi/2,\pi/2)$ and $\phi_{k,l} \in [-\pi/2,\pi/2)$ being the downtilt and the azimuth of the $l$th propagation path of user $k$, respectively;
\begin{equation}\label{Eq:vecp}
{\bf p} (\tau)= {\left[1, e^{j2\pi \triangle f \tau},\ldots,e^{j2\pi (N-1)\triangle f \tau}\right]^T}
\end{equation}
is the delay vector on the OFDM sub-carriers; and $\tau_{k,l}$ is the delay of the $l$th propagation path of user $k$. Here, the steering vector of UPA can be explicitly expressed as
\begin{equation}\label{Eq:ULAstevec}
\begin{aligned}
{\bf a} (\theta,\phi)&={\bf a}_v (\theta)\otimes {\bf a}_h (\theta,\phi), \\
{\bf a}_v (\theta)&= {\left[1,e^{j2\pi \frac{d}{\lambda} \sin\theta}, \ldots, e^{j2\pi (M_v-1)\frac{d}{\lambda} \sin\theta}\right]^{T}},\\
{\bf a}_h (\theta,\phi)&= {\left[1,e^{j2\pi \frac{d}{\lambda} \cos\theta\sin\phi}, \ldots, e^{j2\pi (M_h-1)\frac{d}{\lambda} \cos\theta\sin\phi}\right]^{T}}.
\end{aligned}
\end{equation}
Given the frequency-independent nature of the propagation delays and angles, the baseband channel of user $k$ in the downlink can be modeled as
\begin{equation}\label{Eq:DLChannel}
{\bf h}^{\rm dl}_k=\sum\limits_{l = 1}^{{L_k}} {g^{\rm dl}_{k,l} {\bf a}^T(\theta_{k,l},\phi_{k,l})\otimes {\bf p}^T(\tau_{k,l}) e^{j2\pi (f^{\rm dl}-f^{\rm ul}) \tau_{k,l}}},
\end{equation}
where $g^{\rm dl}_{k,l}$ is the downlink complex gain of the $l$th propagation path in the $k$th user's channel.

\subsection{Efficient Reconstruction of Downlink Channels}

In the downlink of the FDD massive MIMO system, the BS needs to reconstruct the downlink channel for each user before conducting data transmission.
Comparison between \eqref{Eq:ULChannel} and \eqref{Eq:DLChannel} shows that the uplink and downlink share the same downtilts, azimuths, and delays while each path experiences different phase shifts [i.e., the last term in \eqref{Eq:DLChannel}] because of the carrier frequency shift. We also find that $g^{\rm dl}_{k,l}$ may not be equal to $g^{\rm ul}_{k,l}$ when reflection occurs during propagation \cite{Imtiaz2015}.
On the basis of the observation above, the goal of reconstructing the downlink channel is translated to extracting the frequency-independent parameters $\{ \theta_{k,l}, \phi_{k,l}, \tau_{k,l} \}$ and estimating the downlink gains $\{ g^{\rm dl}_{k,l} \}$, for $l=1,\ldots,L_k$ and $k=1,\ldots,K$. The frequency-independent parameters are more efficient to be extracted from the uplink channels than the downlink channels.
We develop the following transceiver framework for downlink channel reconstruction and data transmission. Fig.~\ref{Fig:WorkingProcess} shows the working process of the transceiver, which consists of the following phases:

\begin{figure}
  \centering
  \includegraphics[scale=0.85]{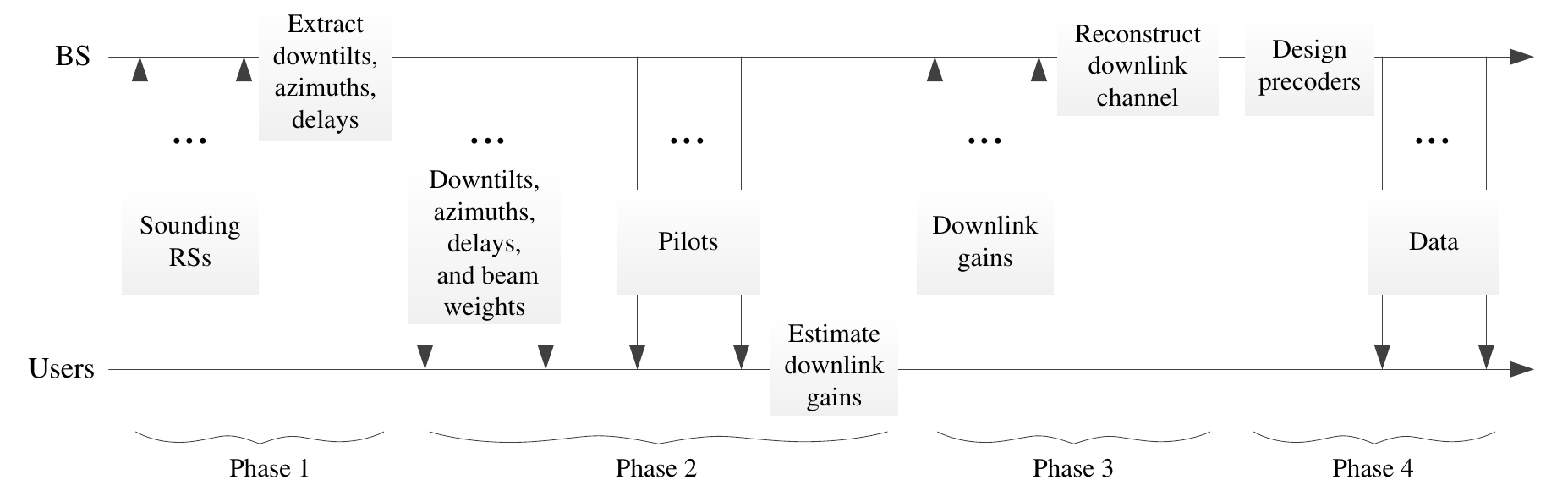}\\
  \caption{Working process of the FDD massive MIMO transceiver.}\label{Fig:WorkingProcess}
\end{figure}

\textbf{ 1) Frequency-independent parameter extraction:}
The users send sounding reference signals (RSs) to the BS. The BS extracts the frequency-independent parameters of the channel from the sounding RSs.

\textbf{ 2) Downlink gain estimation:}
The BS transmits the estimated downtilts, azimuths, delays, and the weights of downlink-training beams to the users and then broadcasts downlink-training pilots.
Each user estimates its downlink gains of the propagation paths on the basis of the known downtilts, azimuths, delays, and beam weights.

\textbf{ 3) Downlink channel reconstruction:}
Users send the estimated downlink gains to the BS. On the basis of the spatial reciprocity, the BS reconstructs the multiuser channel by applying the frequency-independent parameters and the downlink gains in \eqref{Eq:DLChannel}.

\textbf{ 4) Downlink data transmission:}
With full CSI at the transmitter, the BS designs interference-cancellable precoders using ${\bf h}^{\rm dl}_k$, $k=1,\ldots,K$, and maximizes the spatial multiplexing gain by serving all the $K$ users simultaneously in the downlink.


The channel reconstruction through the procedures above is reasonable.
However, the use of a large-scale UPA array and the existence of multiple users bring new challenges:
\begin{itemize}
  \item Three kinds of frequency-independent parameters need to be estimated in Phase 1.
  \item We should estimate downlink gains for multiple users with an acceptable amount of downlink training overhead in Phase 2 to ensure that sufficient time resources are retained for the multiuser data transmission in Phase 4.
\end{itemize}
The two key problems will be solved in the following sections.

\section{Extraction of Frequency-independent Parameters in the Uplink}\label{Sec:NOMP}

In this section, we focus on the first challenge mentioned above and obtain the delays, azimuths, and downtilts of each user channel for massive MIMO-OFDM systems.

\subsection{Uplink RS Model}

\begin{figure}
  \centering
  \includegraphics[scale=0.9]{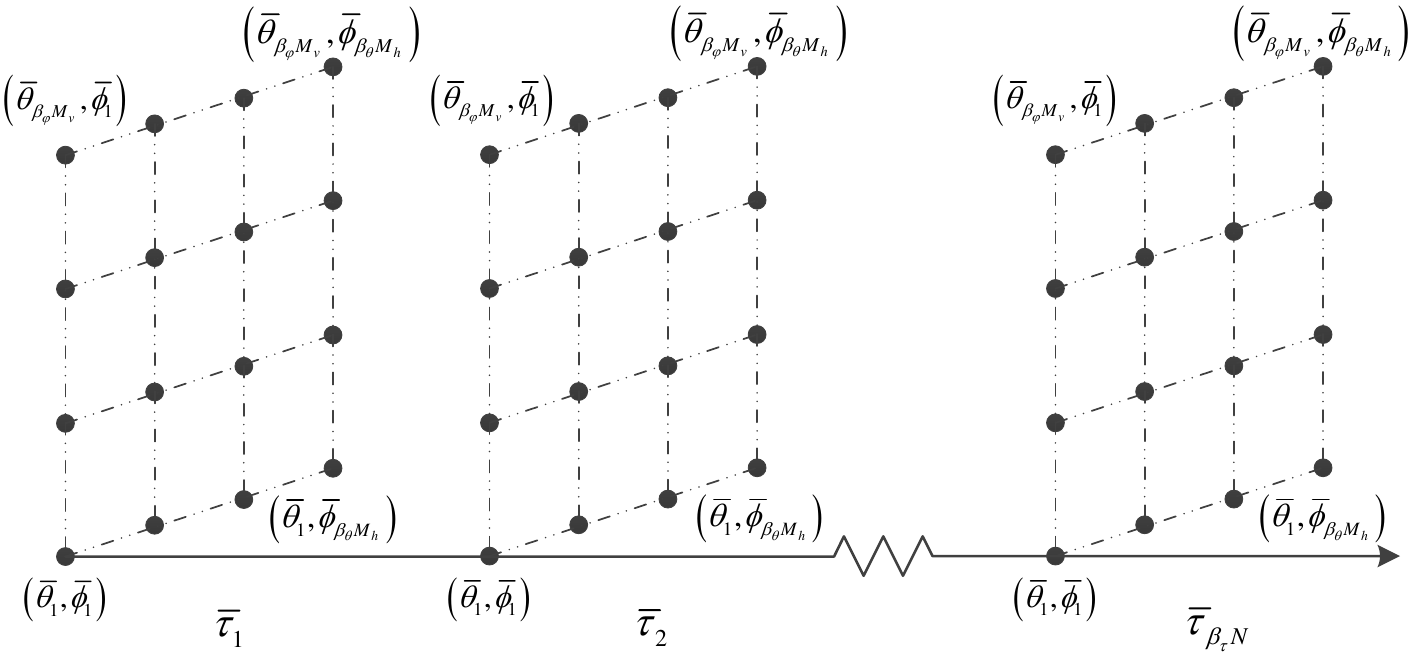}\\
  \caption{Codebook used in the eNOMP algorithm. Each vertical plane represents a lower dimensional sub-codebook that covers the sampled downtilts and azimuths on a sampled delay.}\label{Fig:Codebook}
\end{figure}

During the uplink sounding phase, each user sends sounding RSs to the BS. Sounding RSs from different users are time separated and the BS can distinguish sounding RSs from different users. We assume that the sounding RS from user $k$ occupies the $k$th OFDM symbol in the uplink slot and that all-ones sounding RSs are applied. The received sounding RS at the BS from user $k$ is expressed as
\begin{equation}\label{Eq:SRSmodel}
{\bf y}^{\rm ul}_k = \sum\limits_{l = 1}^{L_k}{g^{\rm ul}_{k,l} {\bf a}(\theta_{k,l},\phi_{k,l}) \otimes {\bf p} (\tau_{k,l})}+{\bf z}^{\rm ul}_k,
\end{equation}
where ${\bf z}^{\rm ul}_k \in \mathbb{C}^{MN\times 1}$ is the additive noise vector on all subcarriers of OFDM symbol $k$ and on all antenna elements at the BS, and each element of ${\bf z}^{\rm ul}_k$ is independent and identically distributed (i.i.d.)~with zero mean and unit variance. We aim to extract $\{\tau_{k,l}, \theta_{k,l}, \phi_{k,l}\}_{l=1,\ldots,L_k}$ from ${\bf y}^{\rm ul}_k$.

This frequency (in delay and angular domain) estimation problem can be solved by utilizing the NOMP algorithm that detects frequencies from the noisy mixture of multiple sinusoids. However, the original NOMP algorithm in \cite{Mamandipoor2016} and the extended NOMP algorithm proposed in \cite{Han2018arxiv} do not cover the case with three types of frequencies to be extracted. Hence, in this study, we further develop an eNOMP algorithm to cater for the massive MIMO-OFDM system.

\subsection{eNOMP for Massive MIMO-OFDM Systems}

eNOMP is an iteration-based algorithm that extracts a new component within each iteration through the e-OMP and e-Newton steps. We will describe the two steps later.
At the end of the $i$th iteration of the eNOMP algorithm, the $i$th component path will be removed from the noisy mixture. If the parameters are precisely estimated, then the $i$th component path will be completely eliminated and the residual noisy mixture will be minimized.
When the algorithm terminates, the number of extracted components equals the number of practical components if each component is accurately estimated.

We give a detailed description of the e-OMP and e-Newton steps in the $i$th iteration of the eNOMP algorithm.
\subsubsection{e-OMP Step}
At the beginning, the residual noisy mixture is expressed as
\begin{equation}\label{Eq:NOMPresidual}
{\bf y}^{\rm ul}_{{\rm r},k}(i) = {\bf y}^{\rm ul}_k - \sum\limits_{l = 1}^{{i-1}}{{\hat g}^{\rm ul}_{k,l} {\bf a}(\hat\theta_{k,l},\hat\phi_{k,l}) \otimes {\bf p} (\hat\tau_{k,l})},
\end{equation}
where $\hat g^{\rm ul}_{k,l}$, $\hat\theta_{k,l}$, $\hat\phi_{k,l}$, and $\hat\tau_{k,l}$ are the gain, downtilt, azimuth, and delay estimated in the $l$th iteration, respectively.
In the e-OMP step, we exhaustively search a pre-defined codebook that occupies the value regions of the downtilt, azimuth, and delay to find the codeword that best matches ${\bf y}^{\rm ul}_{{\rm r},k}(i)$, and use the downtilt, azimuth, and delay included in this codeword as the coarse estimates of downtilt, azimuth, and delay of the extracted $i$th path.

To fit the massive MIMO-OFDM scenario, we design the codebook in accordance with the structure of a component path in \eqref{Eq:SRSmodel}. A codeword is expressed as
\begin{equation}\label{Eq:NOMPcodeword}
{\bf c}(\bar\theta,\bar\phi,\bar\tau) = {\bf a}(\bar\theta,\bar\phi) \otimes {\bf p} (\bar\tau),
\end{equation}
where $\bar\theta \in [-\pi/2,\pi/2)$, $\bar\phi \in [-\pi/2,\pi/2)$, and $\bar\tau \in [0,1/{\triangle f})$ are the downtilt, azimuth, and delay that ${\bf c}$ represents.
The codebook covers the 3D space and the delay domain. Thus, we sample the angles and delays uniformly as
\begin{equation}\label{Eq:NOMPgrids}
\begin{aligned}
\bar\theta &\in \left\{ \bar\theta_1=-\frac{\pi}{2},\bar\theta_2=-\frac{\pi}{2}+\frac{\pi}{\beta_{\theta}M_v},\ldots,\bar\theta_{\beta_{\theta}M_v}=-\frac{\pi}{2}+\frac{(\beta_{\theta}M_v-1)\pi}{\beta_{\theta}M_v}\right\},\\
\bar\phi &\in \left\{ \bar\phi_1=-\frac{\pi}{2},\bar\phi_2=-\frac{\pi}{2}+\frac{\pi}{\beta_{\phi}M_h},\ldots,\bar\phi_{\beta_{\phi}M_h}=-\frac{\pi}{2}+\frac{(\beta_{\phi}M_h-1)\pi}{\beta_{\phi}M_h}\right\},\\
\bar\tau &\in  \left\{ \bar\tau_1=0, \bar\tau_2=\frac{1}{\beta_{\tau}N\triangle f},\ldots,\bar\tau_{\beta_{\tau}N}=\frac{\beta_{\tau}N-1}{\beta_{\tau}N\triangle f}\right\},
\end{aligned}
\end{equation}
where $\beta_{\theta}$, $\beta_{\phi}$, and $\beta_{\tau}$ are the over-sampling rates of the downtilt, azimuth, and delay, respectively. The black circles in Fig.~\ref{Fig:Codebook} are the codewords. Each codeword points to a sampled spatial direction and covers a sampled delay.
The e-OMP step selects the codeword with the maximum projected power from ${\bf y}^{\rm ul}_{{\rm r},k}(i)$, that is,
\begin{equation}\label{Eq:NOMPstep1match}
{\bf c}(\hat\theta_{k,i},\hat\phi_{k,i},\hat\tau_{k,i}) = \mathop{\arg\max}_{(\bar\theta,\bar\phi,\bar\tau)}\frac{|{\bf c}^H(\bar\theta,\bar\phi,\bar\tau) {\bf y}^{\rm ul}_{{\rm r},k}(i)|^2}{\|{\bf c}(\bar\theta,\bar\phi,\bar\tau)\|^2},
\end{equation}
where $\hat\theta_{k,i}$, $\hat\phi_{k,i}$, and $\hat\tau_{k,i}$ are the coarsely estimated downtilt, azimuth, and delay of the $i$th component path.
Then, the gain of the $i$th component path is calculated by
\begin{equation}\label{Eq:NOMPstep1gain}
\hat g^{\rm ul}_{k,i} = \frac{{\bf c}^H(\hat\theta_{k,i},\hat\phi_{k,i},\hat\tau_{k,i}) {\bf y}^{\rm ul}_{{\rm r},k}(i)}{\|{\bf c}(\hat\theta_{k,i},\hat\phi_{k,i},\hat\tau_{k,i})\|^2}.
\end{equation}

Notably, the total number of codewords equals $\beta_{\theta}M_v \times \beta_{\phi}M_h\times\beta_{\tau}N$. Increasing the over-sampling rates helps improve the matching between ${\bf y}^{\rm ul}_{{\rm r},k}(i)$ and ${\bf c}(\hat\theta_{k,i},\hat\phi_{k,i},\hat\tau_{k,i})$ and enhance the precision of the coarse estimates of azimuth, downtilt, and delay. However, this step also multiplies the search time and severely degrades the efficiency, especially when $M_v$, $M_h$, and $N$ are large. Therefore, the over-sampling rates are usually small when we design the codebook for massive MIMO systems.

\subsubsection{e-Newton Step}

Before removing the $i$th component path from the noisy mixture at the final stage of e-OMP, the e-Newton step is applied to tackle the off-grid effect and adjust the estimates toward the real values. Newton's method can successively find better approximations to the roots of a function \cite{Deuflhard2004}. The goal of minimizing the residual noisy mixture is realized by maximizing
\begin{equation}\label{Eq:NOMPS}
S(\theta,\phi,\tau)=2\Re\left\{{\bf y}^{{\rm ul}H}_{{\rm r}}(i) g^{\rm ul}{\bf c}(\theta,\phi,\tau)\right\} - \|g^{\rm ul}{\bf c}(\theta,\phi,\tau)\|^2.
\end{equation}
The e-Newton step is designed to refine the downtilt, azimuth, and delay simultaneously by
\begin{equation}\label{Eq:NOMPstep2}
\left[ \begin{matrix} {\hat\theta'}_{k,i}\\ {\hat\phi'}_{k,i} \\ {\hat\tau'}_{k,i} \end{matrix} \right]= \left[ \begin{matrix} {\hat\theta}_{k,i}\\ {\hat\phi}_{k,i} \\ {\hat\tau}_{k,i} \end{matrix} \right] - {\ddot{\bf S}\left(\hat\theta_{k,i},\hat\phi_{k,i},\hat\tau_{k,i}\right)}^{-1} {\dot{\bf S}\left(\hat\theta_{k,i},\hat\phi_{k,i},\hat\tau_{k,i}\right)},
\end{equation}
where
\begin{equation}\label{Eq:NOMPdotsS}
\dot{\bf S}\left(\theta,\phi,\tau\right)=\left[\begin{matrix} {\frac{\partial S}{\partial\theta}} \\ \frac{\partial S}{\partial\phi} \\ \frac{\partial S}{\partial\tau} \end{matrix} \right], \quad
\ddot{\bf S}\left(\theta,\phi,\tau\right)=\left[ \begin{matrix} \frac{\partial^2 S}{\partial\theta^2} & \frac{\partial^2 S}{\partial\theta\partial\phi} & \frac{\partial^2 S}{\partial\theta\partial\tau}\\ \frac{\partial^2 S}{\partial\phi\partial\theta} & \frac{\partial^2 S}{\partial\phi^2} & \frac{\partial^2 S}{\partial\phi\partial\tau}\\
\frac{\partial^2 S}{\partial\tau\partial\theta} & \frac{\partial^2 S}{\partial\tau\partial\phi} & \frac{\partial^2 S}{\partial\tau^2} \end{matrix} \right].
\end{equation}
In \eqref{Eq:NOMPS}, we regard ${\bf y}_{\rm r}^{{\rm ul}}$ and $g^{\rm ul}$ as constant, and the derivation of $S$ is transformed to the derivation of the codeword $\bf c$. We take the partial derivatives of $S$ versus $\theta$ as the examples. The first-order partial derivative is calculated as
\begin{equation}\label{Eq:1deriStheta}
\frac{{\partial {S}}}{{\partial \theta }}= 2\Re \left\{ {{\bf y}^{{\rm ul}H}_{{\rm r}} g^{\rm ul} \frac{{\partial \mathbf{c}}}{{\partial \theta }}}  - {\left| g^{\rm ul} \right|^2}{ \mathbf{c}^H {\frac{{\partial \mathbf{c}}}{{\partial \theta }}} }\right\}.
\end{equation}
The second-order partial and cross partial derivatives are
\begin{equation}\label{Eq:2deriStheta2}
\frac{{\partial^2 {S}}}{{\partial \theta^2 }}= 2\Re \left\{ \left({\bf y}^{{\rm ul}H}_{\rm r} g^{\rm ul} - \left| g^{\rm ul} \right|^2 \mathbf{c}^H \right)\frac{\partial^2 \mathbf{c}}{\partial \theta^2}\right\} - 2\left\|g^{\rm ul}  \frac{\partial \mathbf{c}}{\partial \theta}\right\|^2,
\end{equation}
and
\begin{equation}\label{Eq:2deriSthetaphi}
\frac{{\partial^2 {S}}}{{\partial \theta \partial \phi}}= 2\Re \left\{ \left({\bf y}^{{\rm ul}H}_{\rm r} g^{\rm ul} - \left| g^{\rm ul} \right|^2 \mathbf{c}^H \right)\frac{\partial^2 \mathbf{c}}{\partial \theta \partial \phi}\right\}- 2\Re \left\{ \left| g^{\rm ul} \right|^2 \frac{\partial \mathbf{c}^H}{\partial \phi} \frac{\partial \mathbf{c}}{\partial \theta}\right\},
\end{equation}
respectively. Other derivatives can be derived in a similar way. From \eqref{Eq:NOMPstep2}, we see that the effectiveness of the e-Newton step greatly depends on the initial values, that is, the coarse estimates obtained in the e-OMP step. If the initial values are not far from the true maximum values, then the e-Newton-refined values are closer to global maxima than local maxima, thereby increasing the precision of the estimates of the delay, azimuth, and downtilt. This is one of the reasons why we use oversampled codebook in the e-OMP step.
The gain of the $i$th component path will be updated through \eqref{Eq:NOMPstep1gain} by replacing $(\hat\theta_{k,i},\hat\phi_{k,i},\hat\tau_{k,i})$ with $({\hat\theta'}_{k,i},{\hat\phi'}_{k,i},{\hat\tau'}_{k,i})$. Thereafter, the refined $i$th component path will be removed from ${\bf y}^{{\rm ul}}_{{\rm r},k}(i)$.

Algorithm~\ref{DNOMP} briefly summarizes the working steps of an iteration in the eNOMP algorithm for the massive MIMO-OFDM system. The stopping criterion of the eNOMP iterations is designed on the basis of the required false alarm rate $P_{\rm fa}$ \cite{Mamandipoor2016} whose typical value range is $10^{-1}-10^{-3}$. The algorithm terminates when
\begin{equation}\label{Eq:NOMPstop}
\|\mathcal{F}\{{\bf y}^{\rm ul}_{{\rm r},k}\}\|^2_{\infty} < \ln(MN)-\ln(-\ln(1-P_{\rm fa})),
\end{equation}
where $\mathcal{F}\{\cdot\}$ represents taking Fourier transformation and $\|\cdot\|_{\infty}$ denotes the infinite norm. Note that if each path is precisely estimated, then the algorithm terminates when all the real paths are extracted and no false alarm happens, as long as the the value of $P_{\rm fa}$ is set reasonably.
If the paths extracted in early iterations are not precise enough, then the value of $P_{\rm fa}$ determines how many paths are detected. The e-NOMP algorithm terminates early if $P_{\rm fa}$ is small, and consequently less paths may be detected. Otherwise, the algorithm stops late when increasing $P_{\rm fa}$, and we obtain more paths. The later obtained paths occupy very little proportion of power of the reconstructed channel and may not exist in the real channel. Their existence helps enhance the global accuracy of the reconstructed channel.
Finally, the BS obtains the estimated frequency-independent parameters of all the user channels, which are denoted by $\{ \hat\tau_{k,l}, \hat\theta_{k,l}, \hat\phi_{k,l}\}$, where ${l=1,\ldots,\hat L_k}$, $k = 1,\ldots,K$.

\begin{algorithm}[t]
\caption{\textbf{Working Steps of An Iteration in eNOMP}}
\textbf{Step 1: New detection.} Coarsely estimate downtilt, azimuth, and delay of a component path by applying the e-OMP step. \\
\textbf{Step 2: Single Newton Refinement.} Apply e-Newton step to the newly detected component path.\\
\textbf{Step 3: Cyclic Newton Refinement.} Cyclicly apply e-Newton steps to all the detected component paths one by one.\\
\textbf{Step 4: Gains Update.} Update the gains of all the detected component paths by utilizing the refined parameters.
\label{DNOMP}
\end{algorithm}

\section{Efficient Downlink Training and Channel Reconstruction}\label{Sec:DLtraining}

The multiuser downlink channel can be reconstructed by estimating the downlink gains for each user one-by-one. However, this procedure will cost considerable downlink resources for training pilots, which is the second key problem mentioned in \emph{Section \ref{Sec:SystemModel}}. In this section, we propose an efficient training strategy that utilizes a small amount of training overhead to obtain accurate estimates of downlink gains for multiple users.

\subsection{Requirement for Successful Estimation}

We analyze the requirement for the successful estimation of downlink gains.
For a certain user in the system, the real spatial parameters are $\{\tau_l,\theta_l,\phi_l\}_{l=1,\dots,L}$ and their estimates are $\{\hat\tau_l,\hat\theta_l,\hat\phi_l\}_{l=1,\dots,\hat L}$. In this subsection, we omit the user index to simplify the expressions. During the downlink gain estimation phase, the BS transmits downlink pilots over $T_p$ successive OFDM symbols. To enhance the receiving power of pilots and improve the estimation accuracy, beamforming is applied on the pilots. The beamforming directions are altered on every OFDM symbol.

We suppose that comb-type all-ones pilots are used and these pilots are sparsely and uniformly inserted in the downlink frequency band. The received pilots on the $t$th OFDM symbol at a certain user are
\begin{equation}\label{Eq:DLpilot}
{\bf y}^{\rm dl}(t) = \sum\limits_{l = 0}^{L-1} \sqrt{P}{g^{\rm dl}_l{\bf a}^T(\theta_l,\phi_l){\bf b}_t} {\bf p}_p(\tau_l)+{\bf z}^{\rm dl}(t),
\end{equation}
where $P$ is the transmit power; ${\bf b}_t \in \mathbb{C}^{M\times1}$ is the training beam used on the $t$th OFDM symbol;
\begin{equation}\label{Eq:vecpp}
{\bf p}_p (\tau)={\left[ e^{j2\pi( f^{\rm dl}-f^{\rm ul} + n_1\triangle f) \tau},\ldots,e^{j2\pi ( f^{\rm dl}-f^{\rm ul} + n_{N_p}\triangle f)\tau}\right]^T}
\end{equation}
describes the delay on $N_p$ downlink subcarriers that are occupied by downlink pilots; and ${\bf z}^{\rm dl}(t)$ is the noise vector on the $t$th OFDM symbol with elements that are i.i.d. with zero mean and unit variance. At the user side, considering that $L$, $\theta_l$, $\phi_l$, and $\tau_l$ have been estimated in the uplink and sent to this user, we replace them by $\hat L$, $\hat\theta_l$, $\hat\phi_l$, and $\hat\tau_l$ in the following derivation. We denote
\begin{equation}\label{Eq:DLTheta}
\Theta^{l,t}={\bf a}^T(\hat\theta_l,\hat\phi_l){\bf b}_t
\end{equation}
to simplify the expressions. By stacking all the received pilots into a large vector, we obtain
\begin{equation}\label{Eq:DLpilotvec}
{\bf y}^{\rm dl}= \sqrt{P} {\bf A} {\bf g}^{\rm dl}+{\bf z}^{\rm dl},
\end{equation}
where
\begin{equation}\label{Eq:StackVecsygz}
{\bf y}^{\rm dl}=\left[\begin{matrix} {\bf y}^{\rm dl}(1) \\ \vdots \\ {\bf y}^{\rm dl}(T_p) \end{matrix} \right], \quad
{\bf g}^{\rm dl}=\left[\begin{matrix} g^{\rm dl}_{1} \\ \vdots \\ g^{\rm dl}_{\hat L} \end{matrix} \right], \quad
{\bf z}^{\rm dl}=\left[\begin{matrix} {\bf z}^{\rm dl}(1) \\ \vdots \\ {\bf z}^{\rm dl}(T_p) \end{matrix} \right]
\end{equation}
are the stacked received pilot, downlink gain, and noise vectors, respectively; and
\begin{equation}\label{Eq:DLmtxA}
{\bf A}=\left[\begin{matrix} {\bf A}(1,1) & \cdots & {\bf A}(1,\hat L) \\ \vdots &\ &\vdots \\ {\bf A}(T_p,1) &\cdots &{\bf A}(T_p,\hat L)  \end{matrix} \right]
\end{equation}
is the coefficient matrix with submatrix
\begin{equation}\label{Eq:DLmtxAsub}
{\bf A} (t,l) = e^{j2\pi (f^{\rm dl}-f^{\rm ul}) \hat\tau_l}\Theta^{l,t}{\bf p}_p(\hat\tau_l).
\end{equation}
From \eqref{Eq:DLpilot}, the least squares (LS) estimate of gains is given by
\begin{equation}\label{Eq:DLestg}
{\hat {\bf g}}^{\rm dl}= \frac{1}{\sqrt{P}} {\bf A}^{\dagger}{\bf y}=\frac{1}{\sqrt{P}}({\bf A}^H{\bf A})^{-1}{\bf A}^H{\bf y}.
\end{equation}

Evidently, the first requirement for the successful estimation of ${\bf g}^{\rm dl}$ is that ${\bf A}^H{\bf A}$ is invertible. Thus, ${\bf A}$ must have full column-rank, that is, ${\rm rank}({\bf A}) = \hat L$. An implied condition for full column-rank is that $N_pT_p \ge \hat L$, which indicates that the number of downlink pilots should be no less than the number of estimated downlink gains.
We apply singular-value decomposition on the coefficient matrix ${\bf A}$ by ${\bf A} = {\bf U \Lambda V}^H$, where ${\bf U}\in \mathbb{C}^{N_p T_p \times N_p T_p}$ and ${\bf V}\in \mathbb{C}^{\hat L \times \hat L}$ are unitary matrices and ${\bf \Lambda} \in \mathbb{C}^{N_pT_p \times \hat L}$ satisfies
\begin{equation}\label{Eq:DLmtxLembda}
{\bf \Lambda}=\left[\begin{matrix} \lambda_1 & \ & {\bf 0} \\ \ &\ddots &\ \\ {\bf 0} &\ & \lambda_{\hat L} \\ {\bf 0} & \cdots & {\bf 0} \end{matrix} \right],
\end{equation}
where $\lambda_1 \ge\ldots\ge\lambda_{\hat L}$ are the singular values. To ensure that ${\bf A}^H{\bf A} = {\bf U \Lambda}^H {\bf\Lambda V}^H$ is invertible, the smallest singular value must hold that $|\lambda_{\hat L}|^2>0$.

We denote $\underline{\bf y}^{\rm dl} = {\bf U}^H{\bf y}^{\rm dl}$, $\underline{\bf g}^{\rm dl} = {\bf V}^H{\bf g}^{\rm dl}$, and $\underline{\bf z}^{\rm dl} = {\bf U}^H{\bf z}^{\rm dl}$. Then, the LS estimation of the downlink gain of the $l$th estimated path is
\begin{equation}\label{Eq:DLSVDgl}
\underline{\hat g}^{\rm dl}_l = \underline{g}^{\rm dl}_l + \frac{\lambda_l^*}{\sqrt{P}|\lambda_l|^2}\underline{z}^{\rm dl}_l.
\end{equation}
where $\underline{g}^{\rm dl}_l$ and $\underline{z}^{\rm dl}_l$ are the $l$th elements of $\underline{\bf g}^{\rm dl}$ and $\underline{\bf z}^{\rm dl}$, respectively. The normalized MSE (NMSE) of the estimated downlink gains is defined as
\begin{equation}\label{Eq:DLSVDglMSEdefine}
{\rm NMSE} = \frac{\mathbb{E}\{ \|\underline{\hat{\bf g}}^{\rm dl} - \underline{\bf g}^{\rm dl}\|^2\}}{ \|\underline{\bf g}^{\rm dl}\|^2 },
\end{equation}
which can be further derived as
\begin{equation}\label{Eq:DLSVDglMSE}
{\rm NMSE} =  \frac{1}{P\|\underline{\bf g}^{\rm dl}\|^2} \sum_{l=1}^{\hat L}\frac{\mathbb{E}\{|\underline{z}^{\rm dl}_l|^2\}}{|\lambda_l|^2}.
\end{equation}
Here, we assume a successful estimation of the downlink gains by evaluating the NMSE.

\begin{assumption}\label{def:NMSE}
Estimation of the downlink gains of a user is successful when the NMSE is below a tolerated error rate $\delta$, where $0<\delta\ll1$.
\end{assumption}

Now, we attempt to explore a deeper insight into \eqref{Eq:DLSVDglMSE}. The unitary characteristic of ${\bf U}$ determines that the new noise vector $\underline{\bf z}^{\rm dl}$ holds the same statistics with ${\bf z}^{\rm dl}$. Thus, the elements of $\underline{\bf z}^{\rm dl}$ are also i.i.d.~with zero mean and unit variance, that is, $\mathbb{E}\{|\underline{z}^{\rm dl}_l|^2\}=1$. According to \cite{Xie2018}, the spatial power difference between the uplink and downlink is small. We use $\|\hat{\bf g}^{\rm ul}\|^2$ to approximate $\|\underline{\bf g}^{\rm dl}\|^2$. Then, the NMSE of the gains can be estimated by
\begin{equation}\label{Eq:DLSVDglMSEnew}
\widehat{{\rm NMSE}} = \frac{1}{P\|\hat{\bf g}^{\rm ul}\|^2} \sum_{l=1}^{\hat L}\frac{1}{|\lambda_l|^2}.
\end{equation}
According to \emph{Assumption \ref{def:NMSE}}, we can predict that the downlink gain estimation will be successful if
\begin{equation}\label{Eq:EstSuccess}
\frac{1}{P\|\hat{\bf g}^{\rm ul}\|^2} \sum_{l=1}^{\hat L}\frac{1}{|\lambda_l|^2}<\delta.
\end{equation}
Subsequently, we formulate the requirement for a successful estimation as follows.

\begin{assumption}\label{def:require}
In multiuser FDD massive MIMO systems, the requirement for successful estimation of downlink gains is that
\begin{equation}\label{Eq:DLSVDglMSEthre}
\sum_{l=1}^{\hat L_k}\frac{1}{|\lambda_{l,k}|^2} < \delta P\|\hat{\bf g}^{\rm ul}_k\|^2
\end{equation}
holds for user $k=1,\ldots,K$, where $\lambda_{l,k}$ is the $l$th singular value of the coefficient matrix ${\bf A}_k$ of user $k$ which is defined in \eqref{Eq:DLmtxA}-\eqref{Eq:DLmtxAsub}.
\end{assumption}

This assumption guides our design of the downlink-training strategy. In the following subsection, we will use \eqref{Eq:DLSVDglMSEthre} to design downlink-training beams ${\bf b}_t, t=1,\ldots,T_p$ for achieving successful estimation of the downlink gains.

\subsection{Spatial Angle Grid}

If only one user exists in the cell, then all the downlink pilots can be dedicated to this user and beamformed to the spatial angles estimated in the uplink. In this condition, $T_p = \hat L$, ${\bf b}_t = {\bf a}^*(\hat\theta_l,\hat\phi_l)/\sqrt{M}$. Consequently, $\Theta^{t,t}$ is equal to its maximum value $\sqrt{M}$. From \eqref{Eq:DLSVDglMSEnew}, we can see that the estimation error is inversely proportional to $|\lambda_1|^2, \ldots, |\lambda_{\hat L}|^2$. If $\Theta^{t,t}\gg\Theta^{l,t}$ for $l\ne t$, then $|\lambda_1|,\ldots,|\lambda_{\hat L}|$ approaches to $\sqrt{M}$. We can obtain precise estimation results with extremely small error.

A simple extension to multiuser scenario is to execute the single user estimation process for each user individually. In particular, $T_p = \sum_{k=1}^K {\hat L_k}$ OFDM symbols will be occupied by the pilots. However, this method will be extremely time consuming when the number of users increases. If the training time exceeds the channel coherent time, that is, $T_p \ge T_c$, then the effectiveness of the reconstruction results will be lost.

To ensure a balance between precision and efficiency, we introduce a spatial angle grid and select beamforming directions from the angles on the grid. The spatial angle grid covers the entire space and uniformly samples the spatial directions. Each spatial angle on the grid corresponds to a downlink-training beam. We have defined a codebook with similar usage in eNOMP algorithm. As shown in Fig.~\ref{Fig:Codebook}, each vertical plane represents a lower dimensional sub-codebook that uniformly samples the entire 3D space. Therefore, we can reuse this sub-codebook as the spatial angle grid.

\emph{Remark}: We set $\beta_{\theta}=\beta_{\phi}=1$ to restrict the quantity of downlink-training beams. Notably, if we use an oversampled grid, the probability that a grid point is shared by more than one user greatly decreases, and the number of selected grid points enlarges. Besides, we do not need to precisely target the the downlink pilots on the spatial directions of the paths.  We can successfully estimate the downlink gains only if \eqref{Eq:DLSVDglMSEthre} is satisfied.

For the $i$th grid point, $i=1,\ldots,M_vM_h$, we denote $(\bar\theta,\bar\phi)_i = (\bar\theta_{i_v},\bar\phi_{i_h})$, where the corresponding downtilt and azimuth are
\begin{equation}\label{Eq:GridAngle}
\begin{aligned}
\bar\theta_{i_v} &= \frac{\pi}{M_v}\left(i_v-\frac{M_v}{2}-1\right), \\ \bar\phi_{i_h} &= \frac{\pi}{M_h}\left(i_h-\frac{M_h}{2}-1\right).
\end{aligned}
\end{equation}
In addition, the corresponding indices of the sampled downtilt and azimuth angles are
\begin{equation}\label{Eq:AngleNumber}
i_v = \Big\lceil \frac{i}{M_h} \Big\rceil, \quad i_h = i-M_h(i_v-1),
\end{equation}
respectively. The $i$th grid point corresponds to the beamforming vector ${\bf a}^*(\bar\theta_{i_v},\bar\phi_{i_h})/\sqrt{M}$.

Each estimated pair of downtilt and azimuth angles is rounded to its ``nearest'' spatial grid point with the largest projected power from the original estimated angle pair. For the $l$th estimated path of user $k$, the projected power on grid point $(\bar\theta,\bar\phi)$ is defined as
\begin{equation}\label{Eq:ProjPower}
\rho^{k,l}(\bar\theta,\bar\phi) = \frac{1}{M} {\left|{\bf a}^T(\hat\theta_{k,l},\hat\phi_{k,l}){\bf a}^*(\bar\theta,\bar\phi)\right|}^2.
\end{equation}
To enhance the receiving power at the user side, we calculate the projected power on all the grid points, select the one with the maximum value
\begin{equation}\label{Eq:SelectGrid}
(\bar\theta^{k,l},\bar\phi^{k,l})=\mathop{\arg\max}_{\tiny \begin{array}{c} \bar\theta=\bar\theta_1,\dots,\bar\theta_{M_v}\\ \bar\phi=\bar\phi_{1},\dots,\bar\phi_{M_h}\end{array}} \rho^{k,l}(\bar\theta,\bar\phi)
\end{equation}
and mark it as the optimal grid point.

The main advantages of using a spatial angle grid are listed as follows. (1) When different users mark the same grid point, the corresponding training beam can be shared by these users simultaneously. As a result, the downlink dedicated pilots are translated to downlink broadcast pilots, and the number of required pilots is decreased. (2) The estimated angle pair is projected to an optimal grid point and the projected power is sufficiently large to guarantee high receiving signal-to-noise ratio (SNR) of the downlink pilots.

\subsection{Downlink-Training Strategy}

After the projected grid points from each user are known, the BS records all the marked grid points, includes them in the selected grid angle set, and obtains $\Phi = \{(\bar\theta,\bar\phi)_{i^1},\ldots, (\bar\theta,\bar\phi)_{i^S}\}$. Then, $S$ OFDM symbols will be occupied by the downlink pilots for the estimation of the downlink gains of each user. In a massive MIMO system, the number of users $K$ is usually large. Accordingly, $S$ is large as well. To ensure that sufficient time is retained for data transmission within the channel coherent time, we should reduce the overhead for training.

\subsubsection{Reasons for the Reduction in Training Overhead}

As mentioned above, we mark an optimal grid point for each estimated angle pair to achieve high receiving SNR. According to \cite{Xie2017}, when the number of antennas is large but not unlimited, the estimated angle pair has considerable projected power on more than one grid angle. Therefore, if the transmit power is sufficiently large, then a suboptimal grid point with the second or third largest projected power is also adequate to ensure the high receiving power of pilots. Moreover, when the projected power on other grid points is sufficiently large and the resulting singular values of $\bf A$ satisfy \eqref{Eq:DLSVDglMSEthre}, the suboptimal options work as well. These observations give the possibility to further reduce the grid points included in $\Phi$.
To better understand our idea, we first discuss two specific cases that can reduce the selected grid points.

\emph{Case 1: Share a common grid point.}
We suppose only one user exists in the cell. The propagation paths are spatially-closed with each other. These paths are projected on different optimal grid points but share a common suboptimal grid point. According to \eqref{Eq:DLmtxA}, if the BS beamforms the downlink pilots only on the direction of this common suboptimal grid point, then the coefficient matrix ${\bf A}$ becomes
\begin{equation}\label{Eq:DLmtxAspec2}
{\bf A}=\left[ \Theta^{1,1}{\bf p}_p(\hat\tau_1), \Theta^{2,1}{\bf p}_p(\hat\tau_2),\ldots,\Theta^{\hat L,1}{\bf p}_p(\hat\tau_{\hat L})  \right].
\end{equation}
Notably, ${\bf p}_p(\hat\tau_l)$ is a vector with completely different elements, and $\hat\tau_l$ is not equal to each other for $l=1,\ldots,\hat L$. The probability that ${\bf A}$ satisfies \eqref{Eq:DLSVDglMSEthre} is high, and the downlink gains can still be successfully estimated. Thus, $\hat L-1$ grid points in $\Phi$ are redundant and can be replaced by a single grid point. In this condition, $|\Phi|$ equals 1 rather than $\hat L$, thereby saving training overhead by $\hat L-1$ OFDM symbols.

\emph{Case 2: Utilize grid points selected by other users.}
In this case, we suppose two users exist, and each user only has one propagation path.
We assume that user 1 can receive the broadcast downlink pilot that is beamformed to the optimal grid point selected by user 2. Therefore, this pilot can be utilized by user 1 as well. If the resulting singular values of $\bf A$ satisfy \eqref{Eq:DLSVDglMSEthre}, then the original optimal grid point of the path of user 1 can be removed from $\Phi$.
In this condition, the estimation accuracy will not be degenerated considerably if a user receives downlink pilots that are beamformed to neither optimal nor suboptimal grid angles.


\begin{figure}
  \centering
  \includegraphics[scale=1]{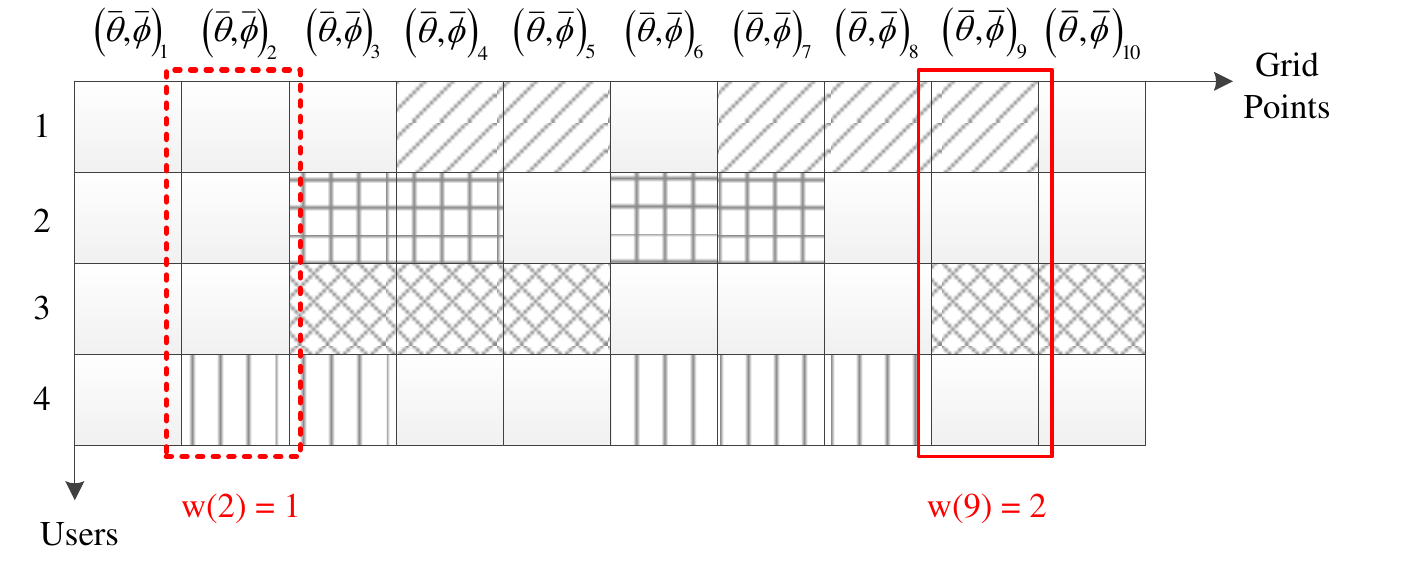}\\
  \caption{Weights of grid points. Grid point $(\bar\theta,\bar\phi)_2$ has the first minimum weight and will thus be abandoned first.}\label{Fig:BeamWeight}
\end{figure}

\subsubsection{Beam Scheduling Scheme}

The two cases indicate that multiple paths can share a common grid point and that each user can utilize the grid points selected by other users. These findings provide reasons for the reduction in training overhead. Based on these findings, we propose a beam scheduling scheme to exclude the redundant grid points in $\Phi$ and only retain the grid points that can be shared by the paths of all the user channels.
The target of the beam scheduling scheme is to minimize the number of selected grid points in $\Phi$:
\begin{equation}\label{Eq:Scheme}
\arg \min |\Phi|, \quad {\rm s.t.} \quad \eqref{Eq:DLSVDglMSEthre} \quad{\rm holds \quad for}\quad k=1,\ldots,K.
\end{equation}

If we exhaustively search all the possible grid angle subsets of $\Phi$ to find the optimum subset, then the searching time will be extremely high. Thus, we follow a similar approach as the greedy method but use it in a reverse way. The redundant grid angles will be identified and excluded from $\Phi$ one by one. According to \emph{Assumption \ref{def:require}}, if \eqref{Eq:DLSVDglMSEthre} holds for all the $K$ users and no more grid points can be excluded in $\Phi$, then we obtain the final set of selected grid points that will be transformed to the downlink-training beams.

We introduce the concept of ``weight'' in the removal of redundant grid angles. If grid point $(\bar\theta,\bar\phi)_i$ is chosen by $\bar K$ users, then its weight is denoted by $w(i)=\bar K$. The large weight of a grid point shows that it is shared by plenty of users. In principle, we keep the grid points that are shared by several users and abandon the grid points that are chosen by only one user. For example, the weights of grid points $(\bar\theta,\bar\phi)_2$ and $(\bar\theta,\bar\phi)_9$ in Fig.~\ref{Fig:BeamWeight} are 1 and 2, respectively. Grid point $(\bar\theta,\bar\phi)_2$ has the first smallest weight and will therefore be abandoned first.

Specifically, we compare $w(i^1),\ldots,w(i^S)$ and reorder $\Phi$ in increasing weight, that is, $w(j^1) \leq w(j^1) \leq \cdots \leq w(j^S)$ and $\Phi = \{(\bar\theta,\bar\phi)_{j^1},\ldots, (\bar\theta,\bar\phi)_{j^S}\}$.
Following the order, we keep excluding the front grid points with small weights until any grid point in $\Phi$ is indispensable to achieve \eqref{Eq:DLSVDglMSEthre}.
The working principle of the beam scheduling scheme is summarized in Algorithm~\ref{scheduling}. We denote the indices of the remaining grid point as $j^{S-T_p+1},\ldots,j^S$. Then, the downlink-training beams are ${\bf b}_t = {\bf a}^*(\bar\theta,\bar\phi)_{j^t}/\sqrt{M}$, $t=S-T_p+1,\ldots,S$.

The complexity of Algorithm~\ref{scheduling} increases in proportion to the number of users. However, if the users are spatially close to one another and their optimal grid points are almost the same, then Algorithm~\ref{scheduling} needs to exclude only a few grid points because the size of the initial grid point set $\Phi$ is small. Under this condition, the complexity of Algorithm 2 does not increase sharply with the increased number of users. However, if the users are spatially separate from one another, the size of the initial grid point set $\Phi$ is large. Then, the complexity increases because the number of grid points with small weights is large.

Before the start of the downlink pilot transmission phase, each selected user is informed with its delays, angles, and the grid angles for the downlink pilots. Thereafter, the BS broadcasts cell-common pilots in successive $T_p$ OFDM symbols.
Each user receives all the pilots and calculates their downlink gains using \eqref{Eq:DLestg}.
The estimated downlink gains are sent back to the BS. By utilizing the uplink-estimated downtilts, azimuths, delays, and the downlink-estimated gains, the BS reconstructs the downlink channel of user $k$ as
\begin{equation}\label{Eq:Reconstruction}
\hat {\bf h}^{\rm dl}_k=\sum\limits_{l = 1}^{{\hat L_k}} {\hat g^{\rm dl}_{k,l} {\bf a}^T(\hat\theta_{k,l},\hat\phi_{k,l})\otimes {\bf p}^T(\hat\tau_{k,l}) e^{j2\pi (f^{\rm dl}-f^{\rm ul}) \hat\tau_{k,l}}}.
\end{equation}
where $\hat g^{\rm dl}_{k,l}$ is the estimated downlink gain of the $l$th path of user $k$. Finally, the BS obtains $\hat{\bf h}^{\rm dl}_1,\ldots,\hat{\bf h}^{\rm dl}_K$.

\begin{algorithm}[t]
\caption{\textbf{Beam Scheduling Strategy}}
\textbf{Require:} Select grid point set $\Phi$\\
\textbf{Initialize:} $\Phi = \{(\bar\theta,\bar\phi)_{i^1},\ldots, (\bar\theta,\bar\phi)_{i^S}\}$ \\
1: Calculate the weight of each grid point in $\Phi$\\
2: Reorder the grid points in $\Phi$ in increasing order and obtain $\Phi = \{(\bar\theta,\bar\phi)_{j^1},\ldots, (\bar\theta,\bar\phi)_{j^S}\}$\\
3: Set ${\rm flag}=1$, $s=1$\\
\textbf{while} ${\rm flag}$ \textbf{do}
\begin{enumerate}
  \item Set $\Phi_{\rm temp} = \{\Phi\backslash (\bar\theta,\bar\phi)_{j^s}\}$
  \item \textbf{for} $k=1,\ldots,K$ \textbf{do}
  \begin{enumerate}
        \item \textbf{if} \eqref{Eq:DLSVDglMSEthre} does not hold for user $k$ when applying $\Phi_{\rm temp}$\\
              \qquad \qquad Set ${\rm flag}=0$, $T_p=S-s+1$, \textbf{break}
  \end{enumerate}
  \item \textbf{if} ${\rm flag}=1$\\
        \qquad \qquad Set $\Phi =\Phi_{\rm temp}$, $s=s+1$
\end{enumerate}
\textbf{Output:} $\Phi$
\label{scheduling}
\end{algorithm}

\subsection{Cost Evaluation}
Here, we evaluate the total cost of the proposed downlink channel reconstruction scheme, including the training and the feedback overhead. We compare the proposed reconstruction with the LMMSE channel estimation method in FDD massive MIMO systems. Table \ref{tab:cost} shows a brief comparison of the costs of these two schemes.

\begin{table}
\caption{Cost Comparison}
\label{tab:cost}
\centering
  \begin{tabular}{|l|c|c|}
  \hline
   & Training & Feedback \\
   & (OFDM symbols) & (complex numbers)\\
  \hline
  LMMSE & $M$ & $MNK$\\
  \hline
  Reconstruction & $T_p$ & $\sum_{k=1}^K {\hat L_k}$\\
  \hline
  \end{tabular}
\end{table}

\subsubsection{Training Overhead}

For the widely used LMMSE channel estimation method, $M$ orthogonal downlink pilots are required to distinguish from the $M$ BS antenna elements. To align with orthogonal pilot design of the downlink-training strategy, the training overhead for LMMSE estimation is $M$ OFDM symbols. As for the proposed reconstruction scheme, the training overhead is $T_p$ OFDM symbols. If we loosen the requirement on the estimation accuracy and increase $\delta$ even in the worst case that $S=M$, most of the beams in the initial $\Phi$ is redundant and will be excluded from $\Phi$. Thus, we obtain that $T_p \ll M$.

\subsubsection{Feedback Overhead}

After the downlink channel is estimated using LMMSE channel estimation method, each user sends the $M$-dimensional complex channel vectors on all the $N$ subcarriers to the BS. Thereafter, the feedback overhead of the LMMSE estimation equals $MNK$ complex numbers. For the proposed downlink channel reconstruction scheme, each user only needs to send back the estimated downlink gains. Thus, the feedback overhead of the proposed scheme is $\sum_{k=1}^K {\hat L_k}$ complex numbers.

\emph{Example}: If $M=128$, $K=10$, and $L_1=\cdots=L_K=6$, then the value range of $T_p$ is [12, 55] according to the numerical results in Section VI. The training overhead of LMMSE estimation is 128 OFDM symbols, and that of the proposed downlink channel reconstruction scheme is 12$-$55 OFDM symbols. If $N=256$, then the LMMSE estimation needs to send back $128\times256\times10$ complex numbers. The proposed reconstruction scheme only feeds back about $6\times10$ complex numbers, far less than that of using LMMSE method.

\section{Multiuser Sum-rate Analysis}\label{Sec:MUscheme}

With the downlink reconstructed channels of all the users, the BS conducts user scheduling and designs data transmission. If the number of users is much smaller than the number of BS antennas, then user scheduling is not necessary, and these users can be simultaneously served. In this section, we analyze the sum-rate based on the reconstructed channels.

During the downlink transmission phase, the BS sends data streams to all the $K$ users simultaneously. To overcome interference, the BS adopts zero-forcing (ZF) precoding before transmitting downlink data. The ZF precoder on downlink subcarrier $n$ is expressed as
\begin{equation}\label{Eq:ZFprecoder}
{\bf W}(n) = \hat{\bf H}^{\dagger}(n) {\bf \Lambda}(n),
\end{equation}
where $\hat{\bf H}(n) \in \mathbb{C}^{K\times M}$ is the reconstructed multiuser channel matrix on downlink subcarrier $n$,
\begin{equation}\label{Eq:esthk}
[\hat{\bf H}(n)]_{k,:}=\hat{\bf h}^{\rm dl}_k(n) = \sum\limits_{l = 1}^{{\hat L_k}} {\hat g^{\rm dl}_{k,l} {\bf a}^T(\hat\theta_{k,l},\hat\phi_{k,l})}e^{j2\pi (f^{\rm dl}-f^{\rm ul}+n\triangle f) \hat\tau_{k,l}},
\end{equation}
and ${\bf \Lambda}(n) = {\rm diag}\{\alpha_1,\ldots,\alpha_K\}$ is the normalization matrix. Here, we adopt uniform power allocation strategy and set
\begin{equation}\label{Eq:ZFpower}
\alpha_k = \frac{1}{\sqrt{K} \|[\hat{\bf H}^{\dagger}(n)]_{:,k}\|}.
\end{equation}

For user $k$, its received data on subcarrier $n$ comprise the target data and the inter-user interference, that is,
\begin{equation}\label{Eq:DLdata1}
r^{\rm dl}_k(n) = \sqrt{P}{\bf h}^{\rm dl}_k(n)\left[{\bf W}(n)\right]_{:,k}d_k(n)+\sum\limits_{j \ne k} \sqrt{P}{{\bf h}^{\rm dl}_k(n)\left[{\bf W}(n)\right]_{:,k}d_j(n)} + z^{\rm dl}_k(n),
\end{equation}
where $r_k(n)$ is the received data by user $k$ on subcarrier $n$, $P$ is the total downlink transmit power, $d_k(n)$ is the transmit data with unit power, and $z_k(n)$ is the noise with unit variance. When applied \eqref{Eq:ZFprecoder}, we rewrite \eqref{Eq:DLdata} as
\begin{equation}\label{Eq:DLdata}
r^{\rm dl}_k(n) = \sqrt{P}\alpha_k{\bf h}^{\rm dl}_k(n)[\hat{\bf H}^{\dagger}(n)]_{:,k}d_k(n)+\sum\limits_{j \ne k} \sqrt{P}\alpha_j{{\bf h}^{\rm dl}_k(n)[\hat{\bf H}^{\dagger}(n)]_{:,k}d_j(n)} + z^{\rm dl}_k(n).
\end{equation}

Considering the cost for the downlink channel reconstruction, the average multiuser sum-rate of this system is calculated as
\begin{equation}\label{Eq:DLsumrate}
R = \left(1-\frac{T_p}{T_c}\right)\frac{1}{N}\sum\limits_{n=1}^{N} \sum\limits_{k=1}^{K} \log_2(1+{\rm SINR}_k(n)),
\end{equation}
where ${\rm SINR}_k(n)$ is the SINR at user $k$ on subcarrier $n$ and can be expressed as
\begin{equation}\label{Eq:SINRdef}
{\rm SINR}_k (n) = \frac{P \alpha_k^2|{\bf h}^{\rm dl}_k (n)[\hat{\bf H}^{\dagger}(n)]_{:,k} |^2}{\sum\limits_{j \ne k} {P\alpha_j^2 |{\bf h}^{\rm dl}_k(n)[\hat{\bf H}^{\dagger}(n)]_{:,j}|^2}+ |z^{\rm dl}_k(n)|^2}.
\end{equation}

Notably, the multiuser sum-rate is proportional to SINR but inversely proportional to $T_p$. We first evaluate the achievable rate of a single user $k$ by analyzing its SINR. To simplify the expressions, we neglect the subcarrier index $n$ in the following derivations. \emph{Theorem \ref{Thrm:SINR}} provides the approximation of the expectation of ${\rm SINR}_k$ when the acceptable error rate equals $\delta$.

\begin{theorem}\label{Thrm:SINR}
The expectation of the SINR at user $k$ approximates
\begin{equation}\label{Eq:ESINRtheo}
\mathbb{E}\{{\rm SINR}_k\} \approx \frac{S_k}{\sum_{j\ne k}I_{k,j}+1},
\end{equation}
where
\begin{equation}\label{Eq:ESINRtheo_Sk}
S_k = \frac{P(1 + \delta\sum_{m=1}^{M} |\left[{\bf H}\right]_{k,m}\left[{\bf H}^{\dagger}\right]_{m,k}|^2)}{K\left(\|\left[{\bf H}^{\dagger}\right]_{:,k}\|^2 + \delta\sum_{j=1}^{K}{\|\left[{\bf H}^{\dagger}\right]_{:,j}\|^2 \sum_{m=1}^{M}{|\left[{\bf H}\right]_{j,m}\left[{\bf H}^{\dagger}\right]_{m,k}|^2} }\right)}
\end{equation}
can be viewed as the expected signal power, and
\begin{equation}\label{Eq:ESINRtheo_Ik}
I_{k,j} = {\frac{P(\delta\sum_{m=1}^{M} |\left[{\bf H}\right]_{k,m}\left[{\bf H}^{\dagger}\right]_{m,j}|^2)}{K\left(\|\left[{\bf H}^{\dagger}\right]_{:,j}\|^2 + \delta\sum_{i=1}^{K}{\|\left[{\bf H}^{\dagger}\right]_{:,i}\|^2 \sum_{m=1}^{M}{|\left[{\bf H}\right]_{i,m}\left[{\bf H}^{\dagger}\right]_{m,j}|^2} }\right)}}
\end{equation}
is viewed as the expected interference power related to user $j$.
\end{theorem}

\begin{proof}
The proof of \emph{Theorem \ref{Thrm:SINR}} can be found in the Appendix.
\end{proof}

From \emph{Theorem \ref{Thrm:SINR}}, we can see that
\begin{equation}\label{Eq:Sk1}
S_k\!=\!\frac{P\left(1 + \delta\sum_{m=1}^{M} |\left[{\bf H}\right]_{k,m}\left[{\bf H}^{\dagger}\right]_{m,k}|^2\right)}{K\|\left[{\bf H}^{\dagger}\right]_{:,k}\|^2\left(1 \!+\! \delta \sum_{m=1}^{M}\!|\left[{\bf H}\right]_{k,m}\left[{\bf H}^{\dagger}\right]_{m,k}|^2\right) \!+\! K\delta X},
\end{equation}
where $X=\sum_{j\ne k}{\|\left[{\bf H}^{\dagger}\right]_{:,j}\|^2 \sum_{m=1}^{M}{|\left[{\bf H}\right]_{j,m}\left[{\bf H}^{\dagger}\right]_{m,k}|^2} } >0$. Since
\begin{equation}\label{Eq:Sk2}
S_k < \frac{\frac{P}{K}\left(1 + \delta\sum_{m=1}^{M}|\left[{\bf H}\right]_{k,m}\left[{\bf H}^{\dagger}\right]_{m,k}|^2\right)}{\|\left[{\bf H}^{\dagger}\right]_{:,k}\|^2\left(1 + \delta{ {\sum_{m=1}^{M}|\left[{\bf H}\right]_{k,m}\left[{\bf H}^{\dagger}\right]_{m,k}|^2} }\right)},
\end{equation}
where the right item can be reduced to $P/{K\|\left[{\bf H}^{\dagger}\right]_{:,k}\|^2}$, we obtain that
\begin{equation}\label{Eq:Sk3}
S_k<\frac{P}{K\|\left[{\bf H}^{\dagger}\right]_{:,k}\|^2},
\end{equation}
which is the expected signal power when $\delta = 0$. The expected signal power degrades if error is acceptable in the LS estimation of the downlink gains. When the value of $\delta$ becomes large, the expected signal power decreases.

Only if $\delta = 0$ and the reconstructed downlink multiuser channel is precise can the interference be completely eliminated. Otherwise, $I_{k,j}>0$, which shows that the interference exists. After rewriting the expression of $I_{k,j}$ by
\begin{equation}\label{Eq:Ik}
I_{k,j} = {\frac{\sum_{m=1}^{M} \frac{P}{K} |\left[{\bf H}\right]_{k,m}\left[{\bf H}^{\dagger}\right]_{m,j}|^2} {\delta^{-1}\|\left[{\bf H}^{\dagger}\right]_{:,j}\|^2 \!+\! \sum_{i=1}^{K}{\|\left[{\bf H}^{\dagger}\right]_{:,i}\|^2 \sum_{m=1}^{M}{|\left[{\bf H}\right]_{i,m}\left[{\bf H}^{\dagger}\right]_{m,j}|^2} }}},
\end{equation}
we find that the interference increases in proportion to $\delta$.

The analytical results above show that, when $\delta$ increases, the signal power decreases and the interference becomes more severe than before, resulting in the considerable degradation in the SINR performance. However, most grids in $\Phi$ are removed by the beam scheduling strategy and the value of $T_p$ is small. Accordingly, the sum-rate increases as well. Therefore, the sum-rate performance remains high even if the requirement on the estimation accuracy is relaxed.

\section{Numerical Results}\label{Sec:Results}

In this section, we evaluate the performance of the downlink channel reconstruction-based FDD massive MIMO transceiver. In this FDD system, $f_{\rm dl}-f_{\rm ul}=300$ MHz, $N=256$, and $\triangle f=75$ kHz. For the UPA at BS, we set $M_v=8$ and $M_h=16$. In each user channel, $T_c=200$ and $L_1=\cdots=L_K=6$. Delays of the paths are randomly distributed in $[0,1/{\triangle f})$. The downtilts and azimuths are randomly distributed in $[-\pi/2,\pi/2)$. Power attenuation occurs during the propagation of a wireless signal, and the total attenuation for each user channel is randomly set within $[0,-10]$ dB. For the e-NOMP algorithm, we set $P_{\rm fa}=10^{-2}$.
The oversampling rates used the in e-OMP step are chosen by jointly considering the resolution of the codebook and the computation complexity. Considering that the values of $N$, $M_h$, and $M_v$ are large, and $N$ is larger than $M_h$ or $M_v$, we set $\beta_{\tau}=1$, and $\beta_{\theta}=\beta_{\phi}=2$ when running the eNOMP algorithm. In the downlink training phase, pilots are uniformly inserted in every four subcarriers.

\subsection{Evaluation of the eNOMP Algorithm}

\begin{figure}
  \centering
  \includegraphics[scale=0.6]{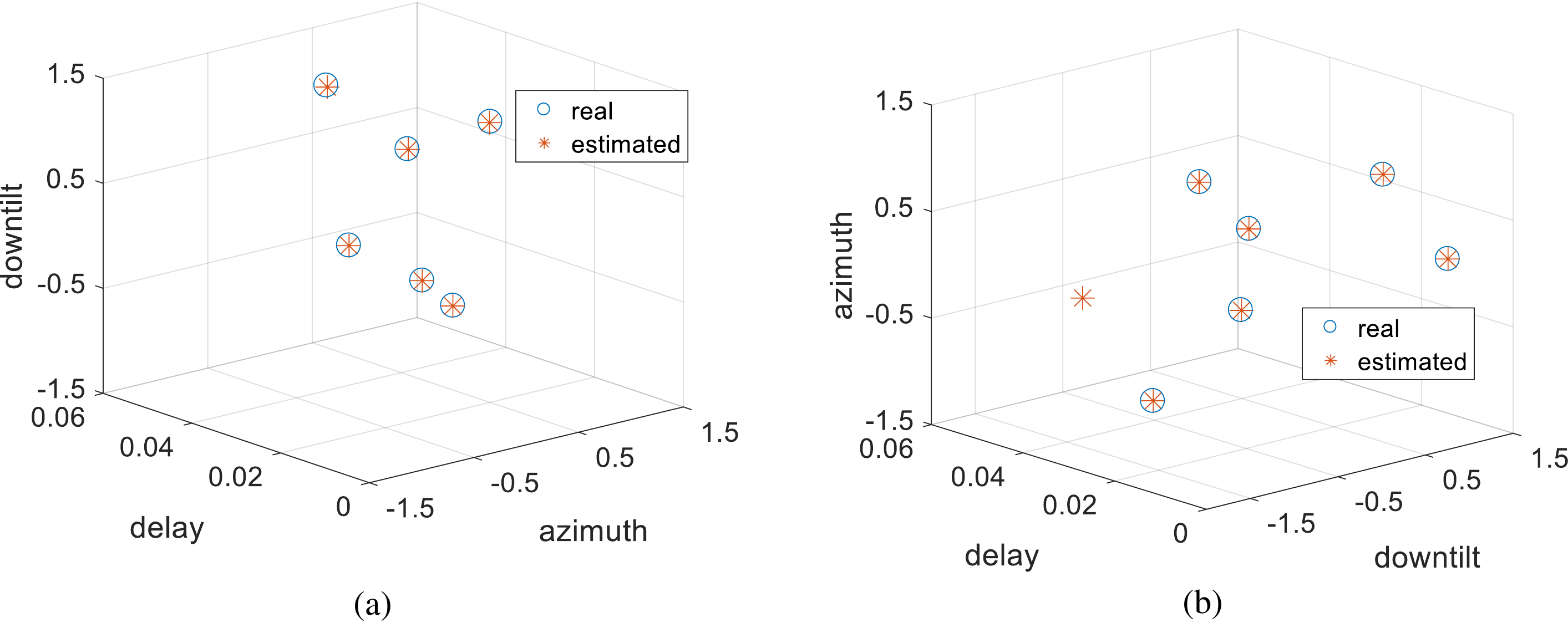}\\
  \caption{Results of two implementations of the eNOMP algorithm. The real frequency-independent parameters are illustrated by circles and their estimates are denoted by stars.}\label{Fig:EstPara3D}
\end{figure}

We evaluate the estimation precision of the eNOMP algorithm by checking if the values of extracted frequency-independent parameters are equal to their real values. Fig.~\ref{Fig:EstPara3D} examines two implementations of the eNOMP algorithm when the channel attenuation and the transmit SNR equal 0 dB. The results are displayed in a 3D coordinate system. The coordinate of each point in the 3D coordinate system is composed of the delay, azimuth, and downtilt. The blue circles denote the real frequency-independent parameters, and the red stars represent their estimates. As shown in Fig.~\ref{Fig:EstPara3D}(a), the estimates coincide with the real values and the number of extracted paths is exactly the real number of paths. Therefore, the eNOMP algorithm can precisely detect each path from the mixture. Fig.~\ref{Fig:EstPara3D}(b) shows that an extra and fake path is detected, which implies that false alarm may occur during the implementation of eNOMP. Considering that the estimated parameters are not sufficiently accurate, the component paths can not be completely eliminated from the mixture. The integration of the residual components will result in a fake component path, which is falsely detected by the algorithm.

\begin{figure}
  \centering
  \includegraphics[scale=0.6]{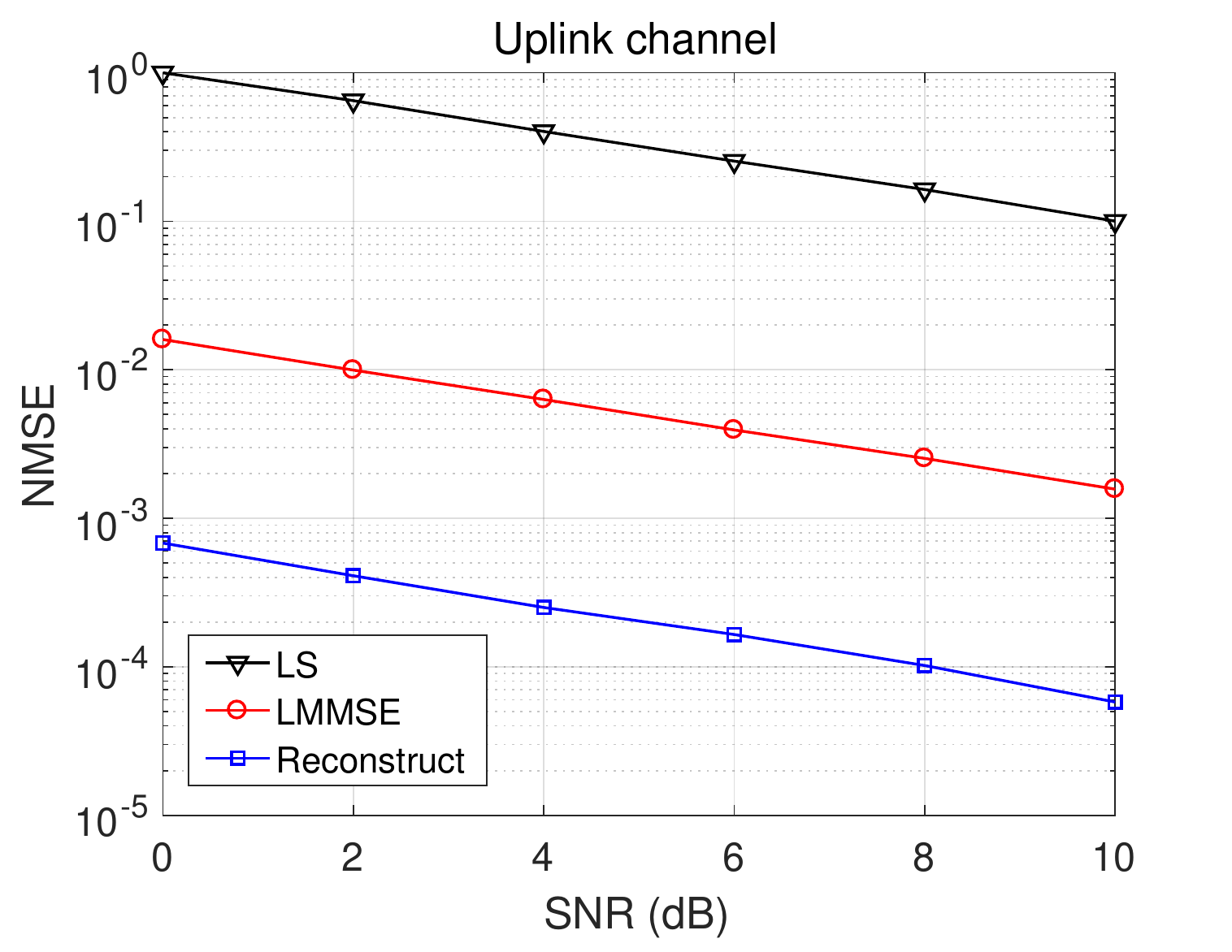}\\
  \caption{NMSE performance of the eNOMP algorithm. The eNOMP-reconstructed uplink channel is more accurate than the LMMSE-estimated uplink channel.}\label{Fig:MSE}
\end{figure}

The occurrence of false alarm is inevitable, but the eNOMP algorithm can still provide a globally accurate reconstruction result. We examine the global accuracy of the eNOMP-based uplink channel reconstruction by evaluating the NMSE, which is calculated as
\begin{equation}\label{Eq:NMSEupchannel}
{\rm NMSE} =\mathbb{E}\left\{ \frac{\|\hat{\bf h}-{\bf h}\|^2}{ \|{\bf h}\|^2}\right\}.
\end{equation}
The classical LS and LMMSE channel estimation methods are introduced as benchmarks. Fig.~\ref{Fig:MSE} compares the NMSE performance of the estimated or reconstructed uplink channels. Here, the wireless channel attenuation is set to 0 dB, and SNR equals the transmit power. As expected, the LS estimated channel has the worst accuracy, and LMMSE improves the performance by a large margin. On the contrary, the eNOMP-based uplink channel reconstruction further reduces the NMSE considerably. When SNR equals 0 dB, eNOMP can obtain $10^{-3}$ NMSE. The value continuously drops with the increase of SNR. These results strongly demonstrate the high global accuracy of the eNOMP algorithm.

\subsection{Evaluation of the Reconstruction-based Transceiver}

We introduce LMMSE channel estimation as a benchmark again and consider its training cost of $M=128$ OFDM symbols to compare the performance of the proposed downlink channel reconstruction-based FDD massive MIMO transceiver. We also evaluate the case when perfect downlink CSI is known at the BS and assume that the training cost is equal to that of the proposed downlink-training strategy. The transmit SNRs in uplink and downlink are equal to 10 dB.

\begin{figure}
  \centering
  \includegraphics[scale=0.6]{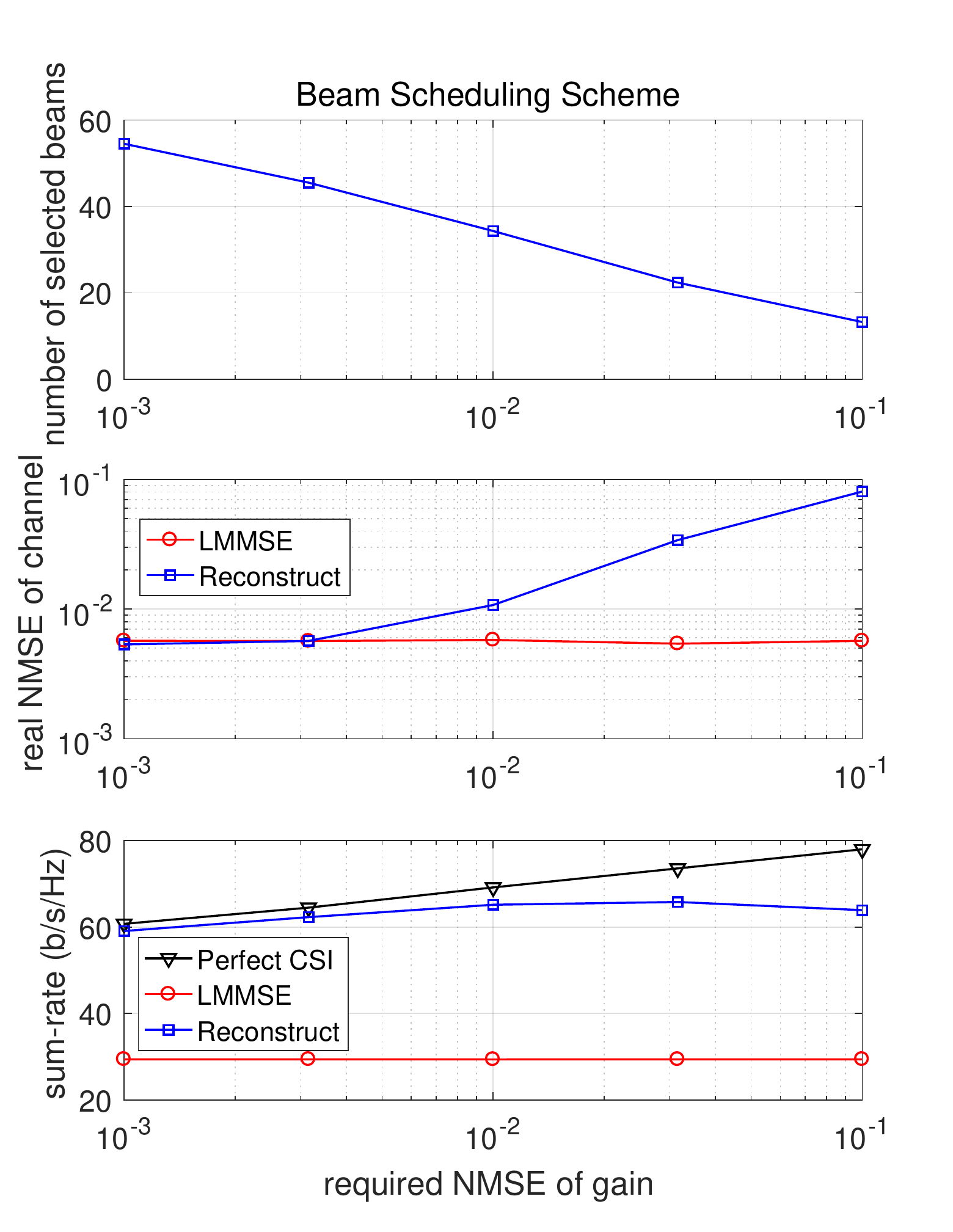}\\
  \caption{Evaluation of the NMSE and multiuser sum-rate performance of the proposed transceiver design with $K=10$. }\label{Fig:rate}
\end{figure}

Fig.~\ref{Fig:rate} evaluates the performance of the transceiver under different requirements on the NMSE of the estimated downlink gains with $K=10$. If we set $\delta=10^{-3}$, then sufficient pilots are available to guarantee the high estimation precision of the downlink gains. Thus, the first sub-figure indicates that seldom grid points are excluded in the selected grid point set. In this condition, NMSE of the reconstructed downlink channel is as small as that of the LMMSE estimated channel when observing the second sub-figure. The real NMSE of the reconstruction scheme is higher than the required NMSE because of a lower bound of the real NMSE, which can be further lowered by enhancing the transmit SNR. Moreover, we find from the third sub-figure that LMMSE estimation has extremely low sum-rate performance due to the large cost of training pilots. This performance gap narrows if the channel's coherence time increases because the time resource for training is not comparable to that of data transmission. The reconstruction scheme has nearly optimal rate performance when compared with the rate of using perfect CSI. Considering that we assume the two methods cost the same amount of training overhead, the proposed scheme can achieve as high SINR as that of using perfect CSI if $\delta$ is small. When $\delta$ increases, the high precision requirement is relaxed gradually. The number of the remaining grid points decreases as well. The reduction in training pilots results in the degradation in accuracy of the estimates. Notably, the real NMSE is nearly the same with the predefined NMSE. The rate performance experiences a slight improvement and reaches the peak when $\delta=10^{-2}$. Therefore, a little sacrifice of NMSE is acceptable, and the reduction in training overhead contributes to the increase in rate. The rate achieved by the reconstruction-based transceiver still approaches that of using perfect CSI. If we further release the NMSE requirement and set $\delta=10^{-1}$, then the rate performance gap between the proposed scheme and using perfect CSI becomes large. Nevertheless, the training cost is decreased to nearly 12 OFDM symbols. As a result, the rate achieved by the proposed scheme is still very high, which demonstrates the efficiency of the proposed scheme.

\section{Conclusion}\label{Sec:Conclusion}

In this paper, we proposed an efficient downlink channel reconstruction-based transceiver for FDD massive MIMO systems. Spatial reciprocity between uplink and downlink was utilized to reduce the training and feedback overhead. We first addressed the problem of extracting downtilts, azimuths, and delays in a 3D massive MIMO-OFDM system by introducing an eNOMP algorithm. Then, we solved the problem of estimating downlink gains for multiple users by utilizing a spatial angle grid and proposing an efficient downlink-training strategy with low overhead. Theoretical analysis revealed the effect of the value of acceptable NMSE on the sum-rate performance. Numerical results proved that downtilts, azimuths, and delays could be precisely estimated using the eNOMP algorithm and that high sum-rate could be achieved by utilizing the reconstructed multiuser downlink channel.

\appendix

According to \eqref{Eq:SINRdef}, the SINR at user $k$ satisfies
\begin{equation}\label{Eq:ESINR}
\mathbb{E}\{{\rm SINR}_k\} \approx \frac{\mathbb{E}\{ P |\alpha_k|^2 |{\bf h}^{\rm dl}_k (n)[\hat{\bf H}^{\dagger}(n)]_{:,k} |^2 \}}{\sum_{j\ne k}\mathbb{E}\{ P |\alpha_j|^2 |{\bf h}^{\rm dl}_k (n)[\hat{\bf H}^{\dagger}(n)]_{:,j} |^2 \}+\mathbb{E}\{|z^{\rm dl}_k|^2\}}.
\end{equation}
Firstly, it is obvious that $\mathbb{E}\{|z^{\rm dl}_k|^2\}=1$. For the signal and interference items, they can be approximated to, respectively,
\begin{equation}\label{Eq:ESINRsig}
\mathbb{E}\{ P |\alpha_k|^2 |{\bf h}^{\rm dl}_k (n)[\hat{\bf H}^{\dagger}(n)]_{:,k} |^2 \} \approx \frac{P\mathbb{E}\{|{\bf h}^{\rm dl}_k (n)[\hat{\bf H}^{\dagger}(n)]_{:,k} |^2\}}{K \mathbb{E}\{\|[\hat{\bf H}^{\dagger}(n)]_{:,k}\|^2\}}
\end{equation}
and
\begin{equation}\label{Eq:ESINRint}
\mathbb{E}\{ P |\alpha_j|^2 |{\bf h}^{\rm dl}_k (n)[\hat{\bf H}^{\dagger}(n)]_{:,j} |^2 \}\approx \frac{P\mathbb{E}\{|{\bf h}^{\rm dl}_k (n)[\hat{\bf H}^{\dagger}(n)]_{:,j} |^2\}}{K \mathbb{E}\{\|[\hat{\bf H}^{\dagger}(n)]_{:,j}\|^2\}}
\end{equation}
when \eqref{Eq:ZFpower} is applied.
Considering the error in the reconstructed multiuser channel, we model the reconstructed multiuser channel by $\hat{\bf H} = {\bf H} + {\bf E}$, where ${\bf H} \in \mathbb{C}^{K\times M}$ is the real downlink channel matrix, $[{\bf H}]_{k,:}={\bf h}^{\rm dl}_k$, and ${\bf E}$ is the reconstruction error with elements that are i.i.d Gaussian with zero mean. We regard the real channel ${\bf H}$ as constant and acquire the expectation of SINR about the reconstruction error ${\bf E}$. According to \eqref{Eq:DLSVDglMSEthre}, the NMSE of $\hat{\bf g}^{\rm dl}_k$ approximates $\delta$, and the NMSE of $\hat{\bf h}^{\rm dl}_k$ approximates $\delta$ as well. Thus, we make the following approximation
\begin{equation}\label{Eq:errorvar}
\mathbb{E}\{| [{\bf E}]_{k,i}|^2\}\approx \delta |[{\bf H}]_{k,i}|^2.
\end{equation}
Besides, $\hat{\bf H}^{\dagger}$ can be Taylor expanded by \cite{Wang2007}
\begin{equation}\label{Eq:HinvTaylor}
\hat{\bf H}^{\dagger} = ({\bf H}+{\bf E})^{\dagger} \approx {\bf H}^{\dagger}-{\bf H}^{\dagger}{\bf E}{\bf H}^{\dagger}.
\end{equation}
Then, the $k$th column of $\hat{\bf H}^{\dagger}$ is expressed as
\begin{equation}\label{Eq:HkinvTaylor}
[\hat{\bf H}^{\dagger}]_{:,k} \approx [{\bf H}^{\dagger}]_{:,k}-{\bf H}^{\dagger}{\bf E}[{\bf H}^{\dagger}]_{:,k}.
\end{equation}
Since
\begin{equation}\label{Eq:ZFcharacter}
{\bf h}^{\rm dl}_k [{\bf H}^{\dagger}]_{:,j} = \left\{ \begin{array}{lr} 1, & j=k,\\ 0, & j \ne k, \end{array} \right.
\end{equation}
we can derive that
\begin{equation}\label{Eq:hkHinvj}
|{\bf h}^{\rm dl}_k[\hat{\bf H}^{\dagger}]_{:,j}|^2 = \left\{ \begin{array}{lr} 1 \!-\! [{\bf H}^{\dagger}]_{:,k}^H[{\bf E}]_{k,:}^H \!-\! [{\bf E}]_{k,:}[{\bf H}^{\dagger}]_{:,k} \!+\! [{\bf H}^{\dagger}]_{:,k}^H[{\bf E}]_{k,:}^H [{\bf E}]_{k,:} [{\bf H}^{\dagger}]_{:,k} , & j=k\\
{[{\bf H}^{\dagger}]}_{:,j}^H [{\bf E}]_{k,:}^H [{\bf E}]_{k,:} [{\bf H}^{\dagger}]_{:,j},& j \ne k \end{array} \right.
\end{equation}
The expectation approximates
\begin{equation}\label{Eq:EhkHinvj}
\mathbb{E}\{|{\bf h}^{\rm dl}_k[\hat{\bf H}^{\dagger}]_{:,j}|^2\}  \approx \left\{ \begin{array}{lr} 1 + \delta\sum_{m=1}^{M} |[{\bf H}]_{k,m}[{\bf H}^{\dagger}]_{m,k}|^2, & j=k\\
\delta\sum_{m=1}^{M} |[{\bf H}]_{k,m}[{\bf H}^{\dagger}]_{m,j}|^2, & j \ne k \end{array} \right.
\end{equation}
when applying \eqref{Eq:errorvar} to \eqref{Eq:hkHinvj}.
Besides,
\begin{equation}\label{Eq:Hinvk2}
\|[\hat{\bf H}^{\dagger}]_{:,k}\|^2 \approx \|[{\bf H}^{\dagger}]_{:,k}\|^2 + [{\bf H}^{\dagger}]_{:,k}^H {\bf E}^H{\bf H}^{\dagger H}{\bf H}^{\dagger}{\bf E} [{\bf H}^{\dagger}]_{:,k} - [{\bf H}^{\dagger}]_{:,k}^H {\bf H}^{\dagger}{\bf E}[{\bf H}^{\dagger}]_{:,k} - [{\bf H}^{\dagger}]_{:,k}^H {\bf E}^H{\bf H}^{\dagger H}[{\bf H}^{\dagger}]_{:,k} ,
\end{equation}
and its expectation satisfies
\begin{equation}\label{Eq:EHinvk2}
\mathbb{E}\{\|[\hat{\bf H}^{\dagger}]_{:,k}\|^2\} \approx \|[{\bf H}^{\dagger}]_{:,k}\|^2 + \delta\sum_{j=1}^{K}{\|[{\bf H}^{\dagger}]_{:,j}\|^2 \sum_{m=1}^{M}{|[{\bf H}]_{j,m}[{\bf H}^{\dagger}]_{m,k}|^2} }.
\end{equation}
By applying \eqref{Eq:EhkHinvj} and \eqref{Eq:EHinvk2} into \eqref{Eq:ESINRsig}, we obtain
\begin{equation}\label{Eq:ES}
\mathbb{E}\{ P |\alpha_k|^2 |{\bf h}^{\rm dl}_k (n)[\hat{\bf H}^{\dagger}(n)]_{:,k} |^2 \}
\approx\frac{ \frac{P}{K}(1 + \delta\sum_{m=1}^{M} |[{\bf H}]_{k,m}[{\bf H}^{\dagger}]_{m,k}|^2)}{\|[{\bf H}^{\dagger}]_{:,k}\|^2 + \delta\sum_{j=1}^{K}{\|[{\bf H}^{\dagger}]_{:,j}\|^2 \sum_{m=1}^{M}{|[{\bf H}]_{j,m}[{\bf H}^{\dagger}]_{m,k}|^2} }}
\end{equation}
which is exactly $S_k$. Similarly,
\begin{equation}\label{Eq:EI}
\mathbb{E}\{ P |\alpha_j|^2 |{\bf h}^{\rm dl}_k (n)[\hat{\bf H}^{\dagger}(n)]_{:,j} |^2 \} \approx \sum_{j\ne k}I_{k,j}.
\end{equation}
Therefore, \eqref{Eq:ESINRtheo} is obtained.

\end{document}